\documentclass[journal,onecolumn,12pt,draftclsnofoot]{IEEEtran}
\usepackage{amsmath, amsthm, amssymb,cite}
\usepackage{graphicx}
\usepackage{color}
\usepackage{subcaption}
\usepackage{fullpage}
\usepackage{url}
\usepackage{todonotes}
\usepackage[bookmarks=false,colorlinks=true,urlcolor=blue,citecolor=blue,linkcolor=blue]{hyperref}
\usepackage{cleveref}
\usepackage{pifont}

\crefname{equation}{}{}
\Crefname{equation}{}{}
\crefname{thm}{theorem}{theorems}
\Crefname{thm}{Theorem}{Theorems}
\crefname{clm}{claim}{claims}
\Crefname{clm}{Claim}{Claims}
\Crefname{coro}{Corollary}{Corollaries}
\Crefname{lem}{Lemma}{Lemmas}
\Crefname{sec}{Section}{Sections}
\crefname{app}{appendix}{appendices}
\Crefname{app}{Appendix}{Appendices}
\crefname{prop}{proposition}{propositions}
\Crefname{prop}{Proposition}{Propositions}
\Crefname{propty}{Property}{Properties}
\crefname{figure}{fig.}{figures}
\Crefname{figure}{Fig.}{Figures}
\crefname{defn}{definition}{definitions}
\Crefname{defn}{Definition}{Definitions}
\crefname{fact}{fact}{facts}
\Crefname{fact}{Fact}{Facts}
\crefname{appendix}{appendix}{appendices}
\Crefname{appendix}{Appendix}{Appendices}
\crefname{algo}{algorithm}{algorithms}
\Crefname{algo}{Algorithm}{Algorithms}
\crefname{algorithm}{algorithm}{algorithms}
\Crefname{algorithm}{Algorithm}{Algorithms}
\crefname{conj}{conjecture}{conjectures}
\Crefname{conj}{Conjecture}{Conjectures}
\crefname{obs}{observation}{observations}
\Crefname{obs}{Observation}{Observations}

\newtheorem{defn}{Definition}
\newtheorem{rem}{Remark}
\newtheorem{lem}{Lemma}
\newtheorem{thm}{Theorem}
\newtheorem{conj}{Conjecture}

\newtheorem{clm}{Claim}
\newtheorem{prop}{Proposition}

\newcommand{\xor}{\oplus}

\newcommand{\xmark}{\ding{55}}

\newcommand{\expec}{\mathbb{E}}
\newtheorem{exmp}{Example}

\title{\huge{On Throughput-Smoothness Trade-offs \\ in Streaming Communication}}
\author{Gauri Joshi,~\IEEEmembership{Student Member,~IEEE,} 
	 Yuval Kochman,~\IEEEmembership{Member,~IEEE,}
  	Gregory W. Wornell,~\IEEEmembership{Fellow,~IEEE} \thanks{This
    work was presented in the INFOCOM workshop on Communication and Networking Techniques for Contemporary Video, Apr 2014, and the International Symposium on Network Coding, Jun 2014.}  \thanks{G.~Joshi and G. W. Wornell are with the Department of Electrical Engineering and Computer Science, Massachusetts Institute of Technology, Cambridge, MA. Y.~Kochman is with the School of Computer Science and Engineering, the Hebrew University of Jerusalem, Israel.
    (E-mail: gauri@mit.edu, yuvalko@cs.huji.ac.il, gww@mit.edu)} }

\begin{document}
\maketitle

\begin{abstract}
Unlike traditional file transfer where only total delay matters, streaming applications impose delay constraints on each packet and require them to be in order. To achieve fast in-order packet decoding, we have to compromise on the throughput. We study this trade-off between throughput and smoothness in packet decoding. We first consider a point-to-point streaming and analyze how the trade-off is affected by the frequency of block-wise feedback, whereby the source receives full channel state feedback at periodic intervals. We show that frequent feedback can drastically improve the throughput-smoothness trade-off. Then we consider the problem of multicasting a packet stream to two users. For both point-to-point and multicast streaming, we propose a spectrum of coding schemes that span different throughput-smoothness tradeoffs. One can choose an appropriate coding scheme from these, depending upon the delay-sensitivity and bandwidth limitations of the application. This work introduces a novel style of analysis using renewal processes and Markov chains to analyze coding schemes. 
\end{abstract}



\section{Introduction}
\label{sec:intro}
\pdfoutput=1
\subsection{Motivation}
A recent report \cite{sandvine_report} shows that $62 \%$ of the Internet traffic in North America comes from real-time streaming applications. Unlike traditional file transfer where only total delay matters, streaming imposes delay constraints on each individual packet. Further, many applications require in-order playback of packets at the receiver. Packets received out of order are buffered until the missing packets in the sequence are successfully decoded. In audio and video applications some packets can be dropped without affecting the streaming quality. However, other applications such as remote desktop, and collaborative tools such as Dropbox \cite{dropbox} and Google Docs \cite{Googledoc} have strict order constraints on packets, where packets represent instructions that need to be executed in order at the receiver. 

Thus, there is a need to develop transmission schemes that can ensure in-order packet delivery to the user, with efficient use of available bandwidth. To ensure that packets are decoded in order, the transmission scheme must give higher priority to older packets that were delayed, or received in error. However, repeating old packets instead of transmitting new packets results in a loss in the overall rate of packet delivery to the user, i.e., the throughput. Thus there is a fundamental trade-off between throughput and in-order decoding delay. 

The throughput loss incurred to achieve smooth in-order packet delivery can be significantly reduced if the source receives feedback about packet losses. Then the source can adapt its future transmission strategy to strike the right balance between old and new packets. We study this interplay between feedback and the throughput-smoothness trade-off. This analysis can help design transmission schemes that achieve the best throughput-smoothness trade-off with a limited amount of feedback. 

Further, if there are more than one users that request a common stream of packets, then there is an additional inter-dependence between the users. Since the users decode different sets of packets depending of their channel erasures, a packet that is innovative and in-order for one user, may be redundant for another. Thus, the source must strike a balance of giving priority to each of the users. We present a framework to analyze such multicast streaming scenarios. 

\subsection{Previous Work}
When there is immediate and error-free feedback, it is well understood that a simple Automatic-repeat-request (ARQ) scheme is both throughput and delay optimal. But only a few papers in literature have analyzed streaming codes with delayed or no feedback. Fountain codes \cite{luby} are capacity-achieving erasure codes, but they are not suitable for streaming because the decoding delay is proportional to the size of the data. Streaming codes without feedback for constrained channels such as adversarial and cyclic burst erasure channels were first proposed in \cite{emin_thesis}, and also extensively explored in \cite{badr_khisti_isit, patil_badr_khisti}. The thesis \cite{emin_thesis} also proposed codes for more general erasure models and analyzed their decoding delay. These codes are based upon sending linear combinations of source packets; indeed, it can be shown that there is no loss in restricting the codes to be linear. 

However, decoding delay does not capture \emph{in order} packet delivery, which is required for streaming applications. This aspect is captured in the delay metrics in \cite{huan_paper} and \cite{gauri_isit_paper}, which consider that packets are played in-order at the receiver. The authors in \cite{huan_paper} analyze the playback delay of real-time streaming for uncoded packet transmission over a channel with long feedback delay. In \cite{gauri_isit_paper, gauri_masters_thesis_2012} we show that the number of interruptions in playback scales $\Theta(\log n)$ for a stream of length $n$. We then proposed codes that minimize the pre-log term for the no feedback and immediate feedback cases. Recent work \cite{mahadaviani_isit_2015} shows that if the stream is available beforehand (non real-time), it is possible to guarantee constant number of interruptions, independent of $n$. This analysis of in-order playback delay blends tools from renewal processes and large deviations with traditional coding theory. These tools were later employed in interesting work; see e.g.\,\cite{cloud_medard_2015, karzand_2015} which also consider in-order packet delivery delay. In this paper we aim to understand how the frequency of feedback about erasures affects the design of codes to ensure smooth point-to-point streaming.

When the source needs to stream the data to multiple users, even the immediate feedback case becomes non-trivial. The use of network coding in multicast packet transmission has been studied in \cite{fragouli_broadcast, delay_control_online_nw_coding, fu_sadeghi_medard, sundar_sadeghi_medard, shah_three_receiver, broadcast_stability}. The authors in \cite{fragouli_broadcast} use as a delay metric the number of coded packets that are successfully received, but do not allow immediate decoding of a source packet. For two users, the paper shows that a greedy coding scheme is throughput-optimal and guarantees immediate decoding in every slot. However, optimality of this scheme has not been proved for three or more users. In \cite{delay_control_online_nw_coding}, the authors analyze decoding delay with the greedy coding scheme in the two user case. However, both these delay metrics do not capture the aspect of in-order packet delivery. 

In-order packet delivery for multicast with immediate feedback is considered in \cite{fu_sadeghi_medard, sundar_sadeghi_medard, shah_three_receiver}. These works consider that packets are generated by a Poisson process and are greedily added to all future coded combinations. Another related work is \cite{broadcast_stability} which also considers Poisson packet generation in a two-user multicast scenario and derives the stability condition for having finite delivery delay. However, in practice the source can use feedback about past erasures to decide which packets to add to the coded combinations, instead of just greedy coding over all generated packets. In this paper we consider this model of source-controlled packet transmission.

\subsection{Our Contributions}
For the point-to-point streaming scenario, also presented in \cite{gauri_infocom_14}, we consider the problem of how to effectively utilize feedback received by the source to ensure in-order packet delivery to the user. We consider block-wise feedback, where the source receives information about past channel states at periodic intervals. In contrast to playback delay considered in \cite{huan_paper} and \cite{gauri_isit_paper}, we propose a delay metric called the smoothness exponent. It captures the burstiness in the in-order packet delivery. In the limiting case of immediate feedback, we can use ARQ and achieve the optimal throughput and smoothness simultaneously. But as the feedback delay increases, we have to compromise on at least one of these metrics. Our analysis shows that for the same throughput, having more frequent block-wise feedback significantly improves the smoothness of packet delivery. This conclusion is reminiscent of \cite{sahai_error_exponent} which studied the effect of feedback on error exponents. We present a spectrum of coding schemes spanning the throughput-smoothness trade-off, and prove that they give the best trade-off within a broad class of schemes for the no feedback, and small feedback delay cases. 

Next, we extend this analysis to the multicast scenario. We focus of two user case with immediate feedback to the source. Since each user may decode a different set of packets depending of the channel erasures, the in-order packet required by one user may be redundant to the other user. Thus, giving priority to one user can cause a throughput loss to the other. We analyze this interplay between the throughput-smoothness trade-offs of the two users using a Markov chain model for in-order packet decoding. This work was presented in part in \cite{gauri_netcod}. 

\subsection{Organization}

The rest of the paper is organized as follows. In Section~\ref{sec:point_to_point} we analyze the effect of block-wise feedback on the throughput-smoothness trade-off in point-to-point streaming. %
In \Cref{sec:multicast} we present a framework for analyzing in-order packet delivery in multicast streaming with immediate feedback. %
\Cref{sec:conclu} summarizes the results and presents future directions. The longer proofs are deferred to the appendix. 

\section{Preliminaries}
\label{sec:sys_model}
\pdfoutput=1
\label{subsec:sys_model}
\subsection{Source and Channel Model} 

The source has a large stream of packets $s_1, s_2, \cdots, s_n$. The encoder creates a coded packet $y_{n} = f(s_1,\,\,s_2 \,\,..s_n)$ in each slot $n$ and transmits it over the channel. The encoding function $f$ is known to the receiver. For example, if $y_{n}$ is a linear combination of the source packets, the coefficients are included in the transmitted packet so that the receiver can use them to decode the source packets from the coded combination. Without loss of generality, we can assume that $y_n$ is a linear combination of the source packets. The coefficients are chosen from a large enough field such that the coded combinations are independent with high probability.

Each coded combination is transmitted to $K$ users $U_1, U_2, \cdots, U_K$. We consider an i.i.d.\ erasure channel to each user such that every transmitted packet is received successfully at user $U_k$ with probability $p_k$, and otherwise received in error and discarded. The erasure events are independent across the users. An erasure channel is a good model when encoded packets have a set of checksum bits that can be used to verify with high probability whether the received packet is error-free. For the single-user case considered in \Cref{sec:point_to_point}, we denote the channel success probability as $p$, without the subscript. 

\subsection{Packet Delivery}
The application at each user requires the stream of packets to be \emph{in order}. Packets received out of order are buffered until the missing packets in the sequence are decoded. We assume that the buffer is large enough to store all the out-of-order packets. Every time the earliest undecoded packet is decoded, a burst of in-order decoded packets is delivered to the application. For example, suppose that $s_1$ has been delivered and $s_3$, $s_4$, $s_6$ are decoded and waiting in the buffer. If $s_2$ is decoded in the next slot, then $s_2$, $s_3$ and $s_4$ are delivered to the application. 

\subsection{Feedback Model}
We consider that the source receives block-wise feedback about channel erasures after every $d$ slots. Thus, before transmitting in slot $kd+1$, for all integers $k \geq 1$, the source knows about the erasures in slots $(k-1)d +1$ to $kd$. It can use this information to adapt its transmission strategy in slot $kd +1$. Block-wise feedback can be used to model a half-duplex communication channel where after every $d$ slots of packet transmission, the channel is reserved for each user to send $d$ bits of feedback to the source about the status of decoding. The extreme case $d= 1$, corresponds to immediate feedback when the source has complete knowledge about past erasures. And when $d \rightarrow \infty$, the block-wise feedback model converges to the scenario where there is no feedback to the source. 

Note that the feedback can be used to estimate $p_i$ for $1 \leq i \leq K$, the success probablities of the erasure channels, when they are unknown to the source. Thus, the coding schemes we propose for $d < \infty$ are universal; they can be used even when the channel quality of unknown to the source.

\subsection{Notions of Packet Decoding}
We now define some notions of packet decoding that aid the presentation and analysis of coding schemes in the rest of the paper. 

\begin{defn}[Innovative Packets]
\label{defn:innov_pkts}
A coded packet is said to be innovative to a user if it is linear independent with respect to the coded packets received by that user until that time.
\end{defn}

\begin{defn}[Seen Packets]
\label{defn:seen_pkts}
The transmitted marks a packet $s_k$ as ``seen" by a user when it knows that the user has successfully received a coded combination that only includes $s_k$ and packets $s_i$ for $1 \leq i < k$. 
\end{defn}

Since each user requires packets strictly in-order, the transmitter can stop including $s_k$ in coded packets when it knows that all the users have seen it. This is because the users can decode $s_k$ once all $s_i$ for $i < k$ are decoded. 

\begin{defn}[Required packet]
\label{defn:reqd_pkt}
The required packet of $U_i$ is its earliest undecoded packet. Its index is denoted by $r_i$. 
\end{defn}
In other words, $s_{r_i}$ is the first unseen packet of user $U_i$. For example, if packets $s_1$, $s_3$ and $s_4$ have been decoded at user $U_i$, its required packet $s_{r_i}$ is $s_2$.

\subsection{Throughput and Delay Metrics}
\label{subsec:perf_metrics}

We now define the metrics for throughput and smoothness in packet delivery. We define them for a single user here, but the definitions directly extend to multiple users. 

\begin{defn}[Throughput]
\label{defn:thpt}
If $I_n$ is the number of packets delivered in-order to a user until time $n$, the throughput $\tau$ is,
\begin{align}
\lim_{n \rightarrow \infty} \frac{I_n}{n} \text{   in probability} .
\end{align}
\end{defn}

The maximum possible throughput is $\tau = p$, where $p$ is the success probability of our erasure channel. The receiver application may require a minimum level of throughput. For example, if applications with playback require $\tau$ to be greater than the playback rate. 

The throughput captures the overall rate at which packets are delivered, irrespective of their delays. If the channel did not have any erasures, packet $s_k$ would be delivered to the user in slot $k$. The random erasures, and absence of immediate feedback about past erasures results in variation in the time at which packets are delivered. We capture the burstiness in packet delivery using the following delay metric.

\begin{defn}[Smoothness Exponent]
Let $D_k$ be in-order decoding delay of packet $s_k$, the earliest time at which all packets $p_1, \cdots p_k$ are decoded by the receiver. The smoothness exponent $\gamma_k^{(s)}$ is defined as the asymptotic decay rate of $D_k$, which is given by
\begin{align}
\gamma_k^{(s)}= - \lim_{n \rightarrow \infty} \frac { \log \Pr(D_k > n)}{n} \label{eqn:mu_def}
\end{align}
\end{defn}

The relation \eqref{eqn:mu_def} can also be stated as $\Pr(D >n) \doteq e^{-n \gamma}$ where $\doteq$ stands for asymptotic equality defined in \cite[Page 63]{thomas_cover}. For simplicity of analysis we define another delay exponent, the inter-delivery defined as follows. \Cref{thm:equiv_exponents} shows the equivalence of the smoothness and the inter-delivery exponent for time-invariant schemes.

\begin{defn}[Inter-delivery Exponent]
\label{defn:inter_deli_exponent}
Let $T_1$ be the first inter-delivery time, that is the first time instant when one or more packets are decoded in-order. The inter-delivery exponent $\lambda$ is defined as the asymptotic decay rate of $T_1$, which is given by
\begin{align}
\lambda= - \lim_{n \rightarrow \infty} \frac { \log \Pr(T_1 > n)}{n} \label{eqn:decay_rate_T1}
\end{align}
\end{defn}

In this paper we focus on time-invariant transmission schemes where the coding strategy is fixed across blocks of transmission, formally defined as follows.

\begin{defn}[Time-invariant schemes]
\label{defn:time_invariant}
A time-invariant scheme is represented by a vector $\mathbf{x} = [x_1, \cdots x_d]$ where $x_i$, for $ 1 \leq i \leq d$, are non-negative integers such that $\sum_{i} x_i = d$. In each block we transmit $x_i$ independent linear combinations of the $i$ lowest-index unseen packets in the stream, for $1 \leq i \leq d$. 
\end{defn}
The above class of schemes is referred to as time-invariant because the vector $\mathbf{x}$ is fixed across all blocks. Note that there is also no loss of generality in restricting the length of the vector $\mathbf{x}$ to $d$. This is because each block can provide only up to $d$ innovative coded packets, and hence there is no advantage in adding more than $d$ unseen packets to the stream in a given block. 

\begin{thm}
\label{thm:equiv_exponents}
For a time-invariant scheme, the smoothness exponent $\gamma_k^{(s)}$ of packet $s_k$ for any $ k \leq \infty$ is equal to $\lambda$, the inter-delivery exponent.
\end{thm}

Thus, in the rest of the paper we study the trade-off between throughput $\tau$ and the inter-delivery exponent $\lambda$.

\section{Point-to-point Streaming}
\label{sec:point_to_point}
\pdfoutput=1
We first consider the extreme cases of immediate feedback and no feedback in Section~\ref{subsec:immediate_feedback} and Section~\ref{subsec:no_feedback} respectively. In Section~\ref{subsec:block_wise_fb} we propose coding schemes for the general case of block-wise feedback after every $d$ slots. 

\subsection{Immediate Feedback}
\label{subsec:immediate_feedback} 

In the immediate feedback $(d=1)$ case, the source has complete knowledge of past erasures before transmitting each packet. We can show that a simple automatic-repeat-request (ARQ) scheme is optimal in both $\tau$ and $\lambda$. In this scheme, the source transmits the lowest index unseen packet, and repeats it until the packet successfully goes through the channel. 

Since a new packet is received in every successful slot, the throughput $\tau =p$, the success probability of the erasure channel. The ARQ scheme is throughput-optimal because the throughput $\tau = p$ is equal to the information-theoretic capacity of the erasure channel \cite{thomas_cover}. Moreover, it also gives the optimal the inter-delivery exponent $\lambda$ because one in-order packet is decoded in every successful slot. To find $\lambda$, first observe that the tail distribution of the time $T_1$, the first inter-delivery time is,
\begin{align}
\Pr(T_1>n) &= (1-p)^{n} 
\end{align}
Substituting this in Definition~\ref{defn:inter_deli_exponent} we get the exponent $\lambda = -\log(1-p)$. Thus, the trade-off for the immediate feedback case is $(\tau, \lambda) = (p,-\log(1-p))$.

From this analysis of the immediate feedback case we can find limits on the range of achievable $(\tau,\lambda)$ for any feedback delay $d$. Since a scheme with immediate feedback can always simulate one with delayed feedback, the throughput and delay metrics $(\tau,\lambda)$ achievable for any feedback delay $d$ must lie in the region $0 \leq \tau \leq p$, and $0 \leq \lambda \leq -\log(1-p)$. 

%

\subsection{No Feedback}
\label{subsec:no_feedback}
Now we consider the other extreme case $(d = \infty)$, corresponding to when there is no feedback to the source. We propose a coding scheme that gives the best $(\tau, \lambda)$ trade-off among the class of full-rank codes, defined as follows. 
\begin{defn}[Full-rank Codes]
\label{defn:full_rank_codes}
In slot $n$ we transmit a linear combination of all packets $s_1$ to $s_{V[n]}$. We refer to $V[n]$ as the transmit index in slot $n$. 
\end{defn} 
\begin{conj}
\label{conj:full_rank}
Given transmit index $V[n]$, there is no loss of generality in including all packets $s_1$ to $s_{V[n]}$.
\end{conj}

We believe this conjecture is true because the packets are required in-order at the receiver. Thus, every packet $s_j$, $j < V[n]$ is required before packet $s_{V[n]}$ and there is no advantage in excluding $s_j$ from the combination. Hence we believe that there is no loss of generality in restricting our attention to full-rank codes. A direct approach to verifying this conjecture would involve checking all possible channel erasure patterns.


\begin{thm}
\label{thm:opt_no_feedback}
The optimal throughput-smoothness trade-off among full-rank codes is $(\tau, \lambda) = (r , D(r \| p))$ for all $0 \leq r < p$. It is achieved by the coding scheme with $V[n] = \lceil rn \rceil$ for all $n$. 
\end{thm}
The term $D(r \| p)$ is the binary information divergence function, which is defined for $0 < p,r < 1$ as
\begin{equation}
D(r \| p ) = r \log \frac{r}{p} + (1-r) \log \frac{1-r}{1-p},
\end{equation}
where $0 \log 0$ is assumed to be $0$. As $r \rightarrow 0$, $D(r \| p)$ converges to $-\log(1-p)$, which is the best possible $\lambda$ as given in Section~\ref{subsec:immediate_feedback}. 

\begin{figure}[t]
\centering
\includegraphics[width=3.3in]{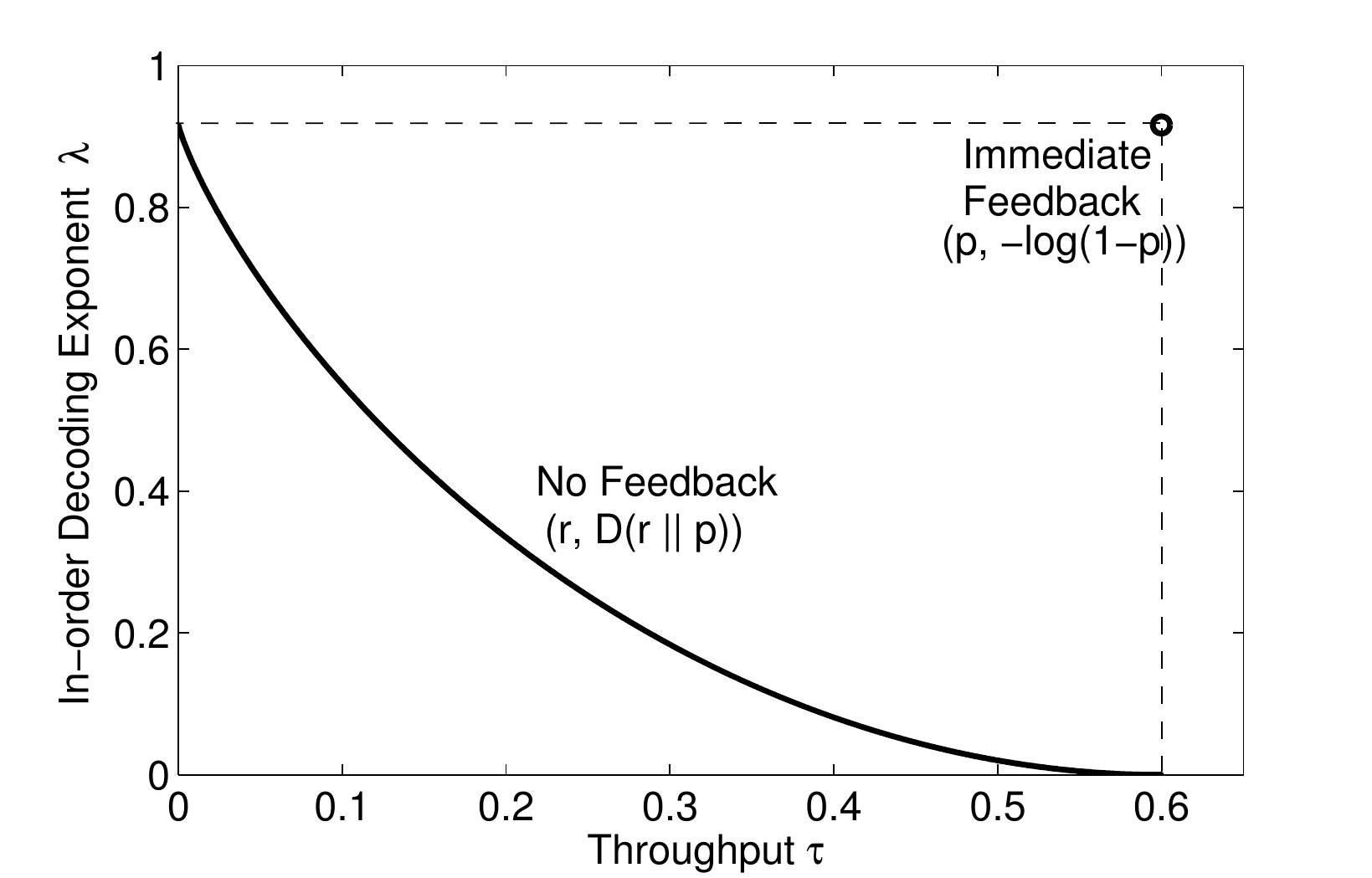}
\caption{The trade-off between inter-delivery exponent $\lambda$ and throughput $\tau$ with success probability $p = 0.6$ for the immediate feedback $(d= 1)$ and no feedback $(d = \infty)$ cases.}
\label{fig:inst_fb_no_fb_tradeoff}
\vspace{-0.2cm}
\end{figure}

Fig.~\ref{fig:inst_fb_no_fb_tradeoff} shows the $(\tau, \lambda)$ trade-off for the immediate feedback and no feedback cases, with success probability $p = 0.6$. The optimal trade-off with any feedback delay $d$ lies in between these two extreme cases.

\subsection{General Block-wise Feedback}
\label{subsec:block_wise_fb}
In Section~\ref{subsec:immediate_feedback} and Section~\ref{subsec:no_feedback} we considered the extreme cases of immediate feedback $(d=1)$ and no feedback $(d = \infty)$ respectively. We now analyze the $(\tau, \lambda)$ trade-off with general block-wise feedback delay of $d$ slots. We restrict our attention to a class of coding schemes called time-invariant schemes, which are defined as follows.



%
Given a vector $\mathbf{x}$, define $p_d$, as the probability of decoding the first unseen packet during the block, and $S_d$ as the number of innovative coded packets that are received during that block. We can express $\tau_{\mathbf{x}}$ and $\lambda_{\mathbf{x}}$ in terms of $p_d$ and $S_d$ as,
\begin{align}
(\tau_{\mathbf{x}}, \lambda_{\mathbf{x}}) &= \left(\frac{\expec[S_d]}{d}, -\frac{1}{d}\log(1-p_d) \right),\label{eqn:lambda_d}
\end{align}
where we get throughput $\tau_{\mathbf{x}}$ by normalizing the $\expec[S_d]$ by the number of slots in the slots. We can show that the probability $\Pr(T_1>kd)$ of no in-order packet being decoded in $k$ blocks is equal $(1-p_d)^k$. Substituting this in \eqref{eqn:decay_rate_T1} we get $\lambda_{\mathbf{x}}$.

\begin{exmp}
Consider the time-invariant scheme $\mathbf{x} = [1, 0, 3, 0]$ where block size $d=4$. That is, we transmit $1$ combination of the first unseen packet, and $3$ combinations of the first $3$ unseen packets. Fig.~\ref{fig:block_wise_exmp} illustrates this scheme for one channel realization. The probability $p_d$ and $\expec[S_d]$ are,
\begin{align}
p_d &= p + (1-p)\binom{3}{3} p^3 (1-p)^0 =  p + (1-p) p^3, \label{eqn:exmp_p_d}\\
\expec[S_d] &= \sum_{i=1}^{3} i \cdot \binom{4}{i} p^i (1-p)^{4-i} + 3p^4 = 4p-p^4, \label{eqn:exmp_E_S_d}
\end{align}
where in \eqref{eqn:exmp_E_S_d}, we get $i$ innovative packets if there are $i$ successful slots for $1 \leq i \leq 3$. But if all $4$ slots are successful we get only $3$ innovative packets. We can substitute \eqref{eqn:exmp_p_d} and \eqref{eqn:exmp_E_S_d} in \eqref{eqn:lambda_d} to get the $(\tau, \lambda)$ trade-off.
\end{exmp}

\begin{figure}[t]
\centering
\includegraphics[scale=0.45]{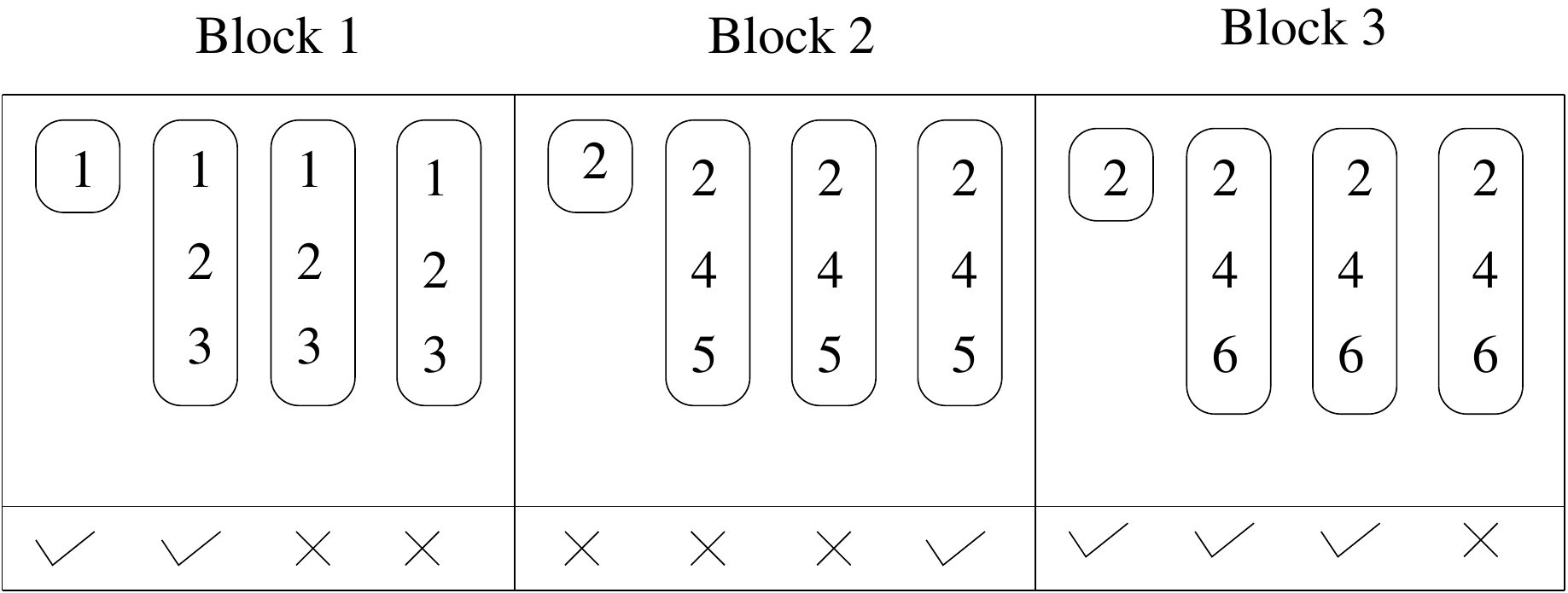}
\caption{Illustration of the time-invariant scheme $\mathbf{x} = [1,0,3,0]$ with block size $d=4$. Each bubble represents a coded combination, and the numbers inside it are the indices of the source packets included in that combination. The check and cross marks denote successful and erased slots respectively. The packets that are ``seen" in each block are not included in the coded packets in future blocks.}
\vspace{-0.1cm}
\label{fig:block_wise_exmp}
\end{figure}

\begin{rem}
\label{rem:uniqueness_of_schemes}
Time-invariant schemes with different $\mathbf{x}$ can be equivalent in terms of the $(\tau, \lambda)$. In particular, given $x_1 \geq 1$, if any $x_i = 0$, and $x_{i+1} = w \geq 1$, then the scheme is equivalent to setting $x_i = 1$ and $x_{i+1} = w-1$, keeping all other elements of $\mathbf{x}$ the same. This is because the number of independent linear combinations in the block, and the probability of decoding the first unseen is preserved by this transformation. For example, $\mathbf{x} = [ 1, 1, 2,0]$ gives the same $(\tau, \lambda)$ as $\mathbf{x} = [ 1,0,3,0]$.
\end{rem}
%
%
%
In Section~\ref{subsec:immediate_feedback} we saw that with immediate feedback, we can achieve $(\tau, \lambda) = (p , -\log(1-p))$. However, with block-wise feedback we can achieve optimal $\tau$ (or $\lambda$) only at the cost of sacrificing the optimality of the other metric. We now find the best achievable $\tau $ (or $\lambda$) with optimal $\lambda$ (or $\tau$).

\begin{clm}[Cost of Optimal Exponent $\lambda$]
\label{clm:cost_of_opt_lambda}
With block-wise feedback after every $d$ slots, and inter-delivery exponent $\lambda = -\log(1-p)$, the best achievable throughput $\tau = (1-(1-p)^d)/d$. 
\end{clm}
\begin{proof}
If we want to achieve $\lambda = -\log(1-p)$, we require $p_d$ in \eqref{eqn:lambda_d} to be equal to $1-(1-p)^d$. The only scheme that can achieve this is $\mathbf{x} = [d, 0, \cdots, 0]$, where we transmit $d$ copies of the first unseen packet. The number of innovative packets $S_d$ received in every block is $1$ with probability $1-(1-p)^d$, and zero otherwise. Hence, the best achievable throughput is $\tau = (1-(1-p)^d)/d$ with optimal $\lambda = -\log(1-p)$.
\end{proof}

This result gives us insight on how much bandwidth (which is proportional to $1/\tau$) is needed for a highly delay-sensitive application that needs $\lambda$ to be as large as possible. 

\begin{clm}[Cost of  Optimal Throughput $\tau$]
\label{clm:cost_of_opt_tau}
With block-wise feedback after every $d$ slots, and throughput $\tau = p$, the best achievable inter-delivery exponent is $\lambda = -\log(1-p)/d$. 
\end{clm}
\begin{proof}
If we want to achieve $\tau = p$, we need to guarantee an innovation packet in every successful slot. The only time invariant scheme that ensures this is $\mathbf{x} = [ 1,0, \cdots, 0, d-1]$, or its equivalent vectors $\mathbf{x}$ as given by Remark~\ref{rem:uniqueness_of_schemes}. With $\mathbf{x} =  [ 1,0, \cdots, 0, d-1]$, the probability of decoding the first unseen packet is $p_d = p$. Substituting this in \eqref{eqn:lambda_d} we get $\lambda =-\log(1-p)/d$, the best achievable $\lambda$ when $\tau = p$.
\end{proof}

Tying back to Fig.~\ref{fig:inst_fb_no_fb_tradeoff}, Claim~\ref{clm:cost_of_opt_lambda} and Claim~\ref{clm:cost_of_opt_tau} correspond to moving leftwards and downwards along the dashed lines from the optimal trade-off $(p, -\log(1-p))$. From Claim~\ref{clm:cost_of_opt_lambda} and Claim~\ref{clm:cost_of_opt_tau} we see that both $\tau$ and $\lambda$ are $\Theta(1/d)$, keeping the other metric optimal. 


%
For any given throughput $\tau$, our aim is to find the coding  scheme that maximizes $\lambda$. We first prove that any convex combination of achievable points $(\tau, \lambda)$ can be achieved. 

\begin{lem}[Combining of Time-invariant Schemes]
\label{lem:interpolate_bw_schemes}
By randomizing between time-invariant schemes $\mathbf{x}^{(i)}$ for $1 \leq i \leq B$, we can achieve the throughput-smoothness trade-off given by any convex combination of the points $(\tau_{\mathbf{x}^{(i)}},\lambda_{\mathbf{x}^{(i)}})$.  
\end{lem}

The proof of \Cref{lem:interpolate_bw_schemes} is deferred to the Appendix. The main implication of Lemma~\ref{lem:interpolate_bw_schemes} is that, to find the best $(\tau,\lambda)$ trade-off, we only have to find the points $(\tau_{\mathbf{x}} ,\lambda_{\mathbf{x}})$ that lie on the convex envelope of the achievable region spanned by all possible $\mathbf{x}$. 

For general $d$, it is hard to search for the $(\tau_{\mathbf{x}}, \lambda_{\mathbf{x}})$ that lie on the optimal trade-off. We propose a set of time-invariant schemes that are easy to analyze and give a good $(\tau, \lambda)$ trade-off. In Theorem~\ref{thm:gen_trade-off} we give the $(\tau, \lambda)$ trade-off for the proposed codes and show that for $d=2$ and $d=3$, it is the best trade-off among all time-invariant schemes.

\begin{defn}[Proposed Codes for general $d$]
\label{defn:gen_d_codes}
For general $d$, we propose using the time-invariant schemes with $x_1 = a$ and $x_{d-a+1} = d-a$, for $a = 1, \cdots d$.
\end{defn}

In other words, in every block of $d$ slots, we transmit the first unseen packet $a$ times, followed by $d-a$ combinations of the first $d-a+1$ unseen packets. These schemes span the $(\tau, \lambda)$ trade-off as $a$ varies from $1$ to $d$, with a higher value of $a$ corresponding to higher $\lambda$ and lower $\tau$. In particular, observe that the $a=d$ and $a=1$ codes correspond to codes given in the proofs of Claim~\ref{clm:cost_of_opt_lambda} and Claim~\ref{clm:cost_of_opt_tau}.

\begin{thm}[Throughput-Smoothness Trade-off for General $d$]
\label{thm:gen_trade-off}
The codes proposed in Definition~\ref{defn:gen_d_codes} give the trade-off points
\begin{align}
(\tau,\lambda) &=\left( \frac{1-(1-p)^a + (d-a)p}{d}, -\frac{a}{d} \log (1-p) \right).\label{eqn:close_to_optimal}
\end{align}
for $a = 1, \cdots d$. 
\end{thm}

\begin{proof}
To find the $(\tau, \lambda)$ trade-off points, we first evaluate $\expec[S_d]$ and $p_d$. With probability $1-(1-p)^a$ we get $1$ innovative packet from the first $a$ slots in a block. The number of innovative packets received in the remaining $d-a$ slots is equal to the number of successful slots. Thus, the expected number of innovative coded packets received in the block is
\begin{align}
\expec[S_d] = 1-(1-p)^a + (d-a) p \label{eqn:E_S_d_gen}
\end{align}
If the first $a$ slots in the block are erased, the first unseen packet cannot be decoded, even if all the other slots are successful. Hence, we have $p_d = 1-(1-p)^a$. 
Substituting $\expec[S_d]$ and $p_d$ in \eqref{eqn:lambda_d}, we get the trade-off in \eqref{eqn:close_to_optimal}.
\end{proof}

By Lemma~\ref{lem:interpolate_bw_schemes}, we can achieve any convex combination of the $(\tau, \lambda)$ points in \eqref{eqn:close_to_optimal}. In Lemma~\ref{lem:d_2_3_tradeoff} we show that for $d=2$ and $d=3$ this is the best trade-off among all time-invariant schemes. 
 
\begin{lem}
\label{lem:d_2_3_tradeoff}
For $d=2$ and $d=3$, the codes proposed in Definition~\ref{defn:gen_d_codes} give the best $(\tau, \lambda)$ trade-off among all time-invariant schemes. 
\end{lem}


\begin{figure}[t]
\centering
\includegraphics[scale=0.45]{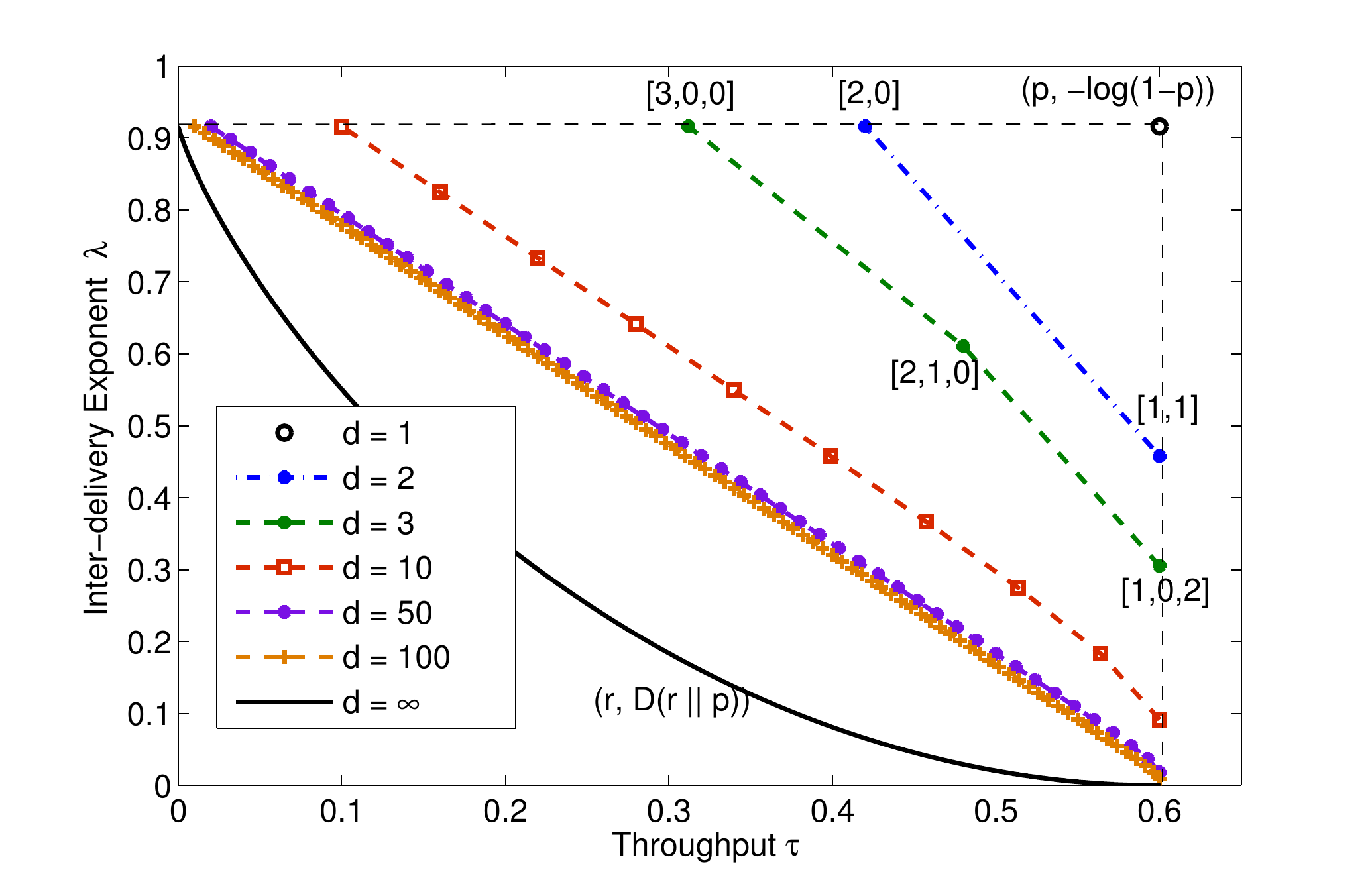}
\caption{The throughput-smoothness trade-off of the suggested coding schemes in Definition~\ref{defn:gen_d_codes} for $p=0.6$ and various values of block-wise feedback delay $d$. The trade-off becomes significantly worse as $d$ increases. The point labels on the $d=2$ and $d=3$ trade-offs are $\mathbf{x}$ vectors of the corresponding codes. }
\label{fig:close_to_optimal_tradeoff}
\vspace{-0.3cm}
\end{figure}

Fig.~\ref{fig:close_to_optimal_tradeoff} shows the trade-off given by \eqref{eqn:close_to_optimal} for different values of $d$. We observe that the trade-off becomes significantly worse as $d$ increases. Thus we can imply that frequent feedback to the source is important in delay-sensitive applications to ensure fast in-order delivery of packets. As $d \rightarrow \infty$, and $a = \alpha d$, the trade-off converges to $( (1-\alpha)p, -\alpha \log(1-p))$ for $ 0 \leq \alpha \leq 1$, which is the line joining $(0, -\log(1-p))$ and $(p, 0)$. It does not converge to the $(r, D(r \| p ))$ curve without feedback because we consider that $d$ goes to infinity slower than the $n$ used to evaluate the asymptotic exponent of $\Pr(T_1 > n)$.

By Lemma~\ref{lem:d_2_3_tradeoff} the proposed codes give the best trade-off among all time-invariant schemes. Numerical results suggest that even for general $d$ these schemes give a trade-off that is close to the best trade-off among all time-invariant schemes. 


Thus, in this section we analyzed how block-wise feedback affects the trade-off between throughput $\tau$ and inter-delivery exponent $\lambda$, which measures the burstiness in-order delivery in streaming communication. Our analysis gives us the insight that frequent feedback is crucial for fast in-order packet delivery. Given that feedback comes in blocks of $d$ slots, we present a spectrum of coding schemes that span different points on the $(\tau, \lambda)$ trade-off. Depending upon the delay-sensitivity and bandwidth limitations of the applications, these codes provide the flexibility to choose a suitable operating point on trade-off. 


\section{Multicast Streaming}
\label{sec:multicast}
\pdfoutput=1
In this section we move to the multicast streaming scenario. Since each user decodes a different set of packets, the in-order packet required by one user may be redundant to the other user, causing the latter to lose throughput if that packet is transmitted. In Section~\ref{subsec:struct_codes} we identify the structure of coding schemes required to simultaneously satisfy the requirements of multiple users. In Section~\ref{subsec:fixed_primary} we use this structure to find the best coding scheme for the two user case, where one user is always given higher priority. In Section~\ref{subsec:general_trade-off} we generalize this scheme to allow tuning the level of priority given to each user. The analysis of both these cases is based on a new Markov chain model of packet decoding. 

For this section, we focus on the case where each user provides immediate feedback to the source $(d=1)$. 
\begin{rem}
For the no feedback case $(d=\infty)$ we can extend Theorem~\ref{thm:opt_no_feedback} to show that the optimal throughput-smoothness trade-off for user $U_k$, $k = 1,2, \cdots K$ among full-rank codes is $(\tau_k, \lambda_k) = (r, D(r \| p_k))$, if $0 \leq r \leq p_k$. If $r > p_k$ then $\lambda_k = 0$ for user $U_k$. Since we are transmitting a common stream, the rate $r$ of adding new packets is same for all users.
\end{rem}

The general $d$ case is hard to analyze and open for future work. 

%
\subsection{Structure of Coding Schemes}
\label{subsec:struct_codes}
The best possible trade-off is $(\tau_i, \lambda_i) = (p_i, -\log(1-p_i))$, and it can be achieved when there is only one user, and the source uses a simple Automatic-repeat-request (ARQ) protocol where it keeps retransmitting the earliest undecoded packet until that packet is decoded. In this paper our objective is to design coding strategies to maximize $\tau$ and $\lambda$ for the two user case. For two or more users we can show that it is impossible to achieve the optimal trade-off $(\tau, \lambda) = (p_i, -\log(1-p_i))$ simultaneously for all users. We now present code structures that maximize throughput and inter-delivery exponent of the users. 
%
\begin{clm}[Include only Required Packets]
\label{clm:reqd_pkts_only}
In a given slot, it is sufficient for the source to transmit a combination of packets $s_{r_i}$ for $i \in \mathcal{I}$ where $\mathcal{I}$ is some subset of $\{1,2, \cdots K\}$. 
\end{clm}

\begin{proof}[Proof]
Consider a candidate packet $s_c$ where $c \neq r_i$ for any $ 1 \leq i \leq K$. If $c < r_i$ for all $i$, then $s_c$ has been decoded by all users, and it need not be included in the combination. For all other values of $c$, there exists a required packet $s_{r_i}$ for some $i \in \{1,2, \cdots K\}$ that, if included instead of $s_c$, will allow more users to decode their required packets. Hence, including that packet instead of $s_c$ gives a higher exponent $\lambda$.
\end{proof}

\begin{clm}[Include only Decodable Packets]
\label{clm:inst_dec_only}
If a coded combination already includes packets $s_{r_i}$ with $ i \in \mathcal{I}$, and $U_j$, $j \notin I$ has not decoded all $s_{r_i}$ for $ i \in \mathcal{I}$, then a scheme that does not include $s_{r_j}$ in the combination gives a better throughput-smoothness trade-off than a scheme that does.
\end{clm}

\begin{proof}[Proof]
If $U_j$ has not decoded all $s_{r_i}$ for $ i \in \mathcal{I}$, the combination is innovative but does not help decoding an in-order packet, irrespective of whether $s_{r_j}$ is included in the combination. However, if we do not include packet $s_{r_j}$, $U_j$ may be able to decode one of the packets $s_{r_i}$, $i \in \mathcal{I}$, which can save it from out-of-order packet decoding in a future slot. Hence excluding $s_{r_j}$ gives a better throughput-smoothness trade-off.
\end{proof}

\begin{exmp}
Suppose we have three users $U_1$, $U_2$, and $U_3$. User $U_1$ has decoded packets $s_1$, $s_2$, $s_3$ and $s_5$, user $U_2$ has decoded $s_1$, $s_3$, and $s_4$, and user $U_3$ has decoded $s_1$, $s_2$, and $s_5$. The required packets of the three users are $s_4$, $s_2$ and $s_3$ respectively. By Claim~\ref{clm:reqd_pkts_only}, the optimal scheme should transmit a linear combination of one or more of these packets. Suppose we construct combination of $s_4$ and $s_2$ and want to decide whether to include $s_3$ or not. Since user $U_3$ has not decoded $s_4$, we should not include $s_3$ as implied by Claim~\ref{clm:inst_dec_only}. 
\end{exmp}

The choice of the initial packets in the combination is governed by a priority given to each user in that slot. Claims~\ref{clm:reqd_pkts_only} and~\ref{clm:inst_dec_only} imply the following code structure for the two user case.

\begin{prop}[Code Structure for the Two User Case]
\label{prop:two_user_code_struct}
Every achievable trade-off between throughput and inter-delivery exponent can be obtained by a coding scheme where the source transmits $s_{r_1}$, $s_{r_2}$ or the exclusive-or, $s_{r_1} \xor s_{r_2}$ in each slot. It transmits $s_{r_1} \xor s_{r_2}$ if and only if $r_1 \neq r_2$, and $U_1$ has decoded $s_{r_2}$ or $U_2$ has decoded $s_{r_1}$.
\end{prop}

In the rest of this section we analyze the two user case and focus on coding schemes as given by Proposition~\ref{prop:two_user_code_struct}.

\begin{figure}[t]
\small
\begin{center}
\begin{tabular}{ |c|c|c|c|}
  \hline
   Time & Sent &  $U_1$ & $U_2$\\
  \hline
     1 	&  $s_1$	& 	$s_1$ & \xmark \\ 
     2 	&  $s_2$	& 	\xmark & $s_2$ \\ 
     3 	&  $s_1 \xor s_2$	& 	$s_2$ & $s_1$ \\
     4 	&  $s_3$	& $s_3$ &\xmark\\
     5 	&  $s_4$	& $s_4$ &$s_4$ \\
   \hline
\end{tabular}
\caption{Illustration of the optimal coding scheme when the source always give priority to user $U_1$. The third and fourth columns show the packets decoded at the two users. Cross marks indicate erased slots for the corresponding user.\label{fig:fixed_primary_eg}}
\end{center}
\vspace{-0.65cm}
\end{figure}

\subsection{Optimal Performance for One of Two Users}
\label{subsec:fixed_primary}

In this section we consider that the source always gives priority to one user, called the primary user. We determine the best achievable throughput-smoothness trade-off for a secondary user that is ``piggybacking" on such a primary user. For simplicity of notation, let $a \triangleq p_1 p_2$, $b \triangleq p_1 (1-p_2)$, $c \triangleq(1-p_1) p_2$ and $d \triangleq (1-p_1) (1-p_2)$, the probabilities of the four possible erasure patterns. 

Without loss of generality, suppose that $U_1$ is the primary user, and $U_2$ is the secondary user. Recall that ensuring optimal performance for $U_1$ implies achieving $(\tau_1, \lambda_1) = (p_1, -\log(1-p_1))$. While ensuring this, the best throughput-smoothness trade-off for user $U_2$ is achieved by the coding scheme given by Claim~\ref{clm:fixed_prim} below.  

\begin{clm}[Optimal Coding Scheme]
\label{clm:fixed_prim}
A coding scheme where the source transmits $s_{r_1} \xor s_{r_2}$ if $U_2$ has already decoded $s_{r_1}$, and otherwise transmits $s_{r_1}$, gives the best achievable $(\tau_2, \lambda_2)$ trade-off while ensuring optimal $(\tau_1, \lambda_1)$.
\end{clm}
\begin{proof}
Since $U_1$ is the primary user, the source must include its required packet $s_{r_1}$ in every coded combination. By Proposition~\ref{prop:two_user_code_struct}, if the source transmits $s_{r_1} \xor s_{r_2}$ if $U_2$ has already decoded $s_{r_1}$, and transmits $s_{r_1}$ otherwise, we get the best achievable throughput-smoothness trade-off for $U_2$.
\end{proof}
Fig.~\ref{fig:fixed_primary_eg} illustrates this scheme for one channel realization. 

Packet decoding at the two users with the scheme given by Claim~\ref{clm:fixed_prim} can be modeled by the Markov chain shown in Fig.~\ref{fig:fixed_primary_markov}.  The state index $i$ can be expressed in terms of the number of gaps in decoding of the users, defined as follows. 

\begin{defn}[Number of Gaps in Decoding]
The number of gaps in $U_i$'s decoding is the number of undecoded packets of $U_i$ with indices less than $r_{\max} = \max_i r_i$. 
\end{defn}

In other words, the number of gaps is the amount by which a user $U_i$ lags behind the user that is leading the in-order packet decoding. The state index $i$, for $i \geq -1$ is equal to the number of gaps in decoding at $U_2$, minus that for $U_1$. Since the source gives priority to $U_1$, it always has zero gaps in decoding, except when there is a $c = p_2(1-p_1)$ probability erasure in state $0$, which causes the system goes to state $-1$. 
The states $i'$ for $i\geq 1$ are called ``advantage" states and are defined as follows.
\begin{defn}[Advantage State]
\label{defn:adv_state}
 The system is in an advantage state when $r_1 \neq r_2$, and $U_2$ has decoded $s_{r_1}$ but $U_1$ has not. 
\end{defn}
 \begin{figure}[t]
\centering
\includegraphics[width=2.5in]{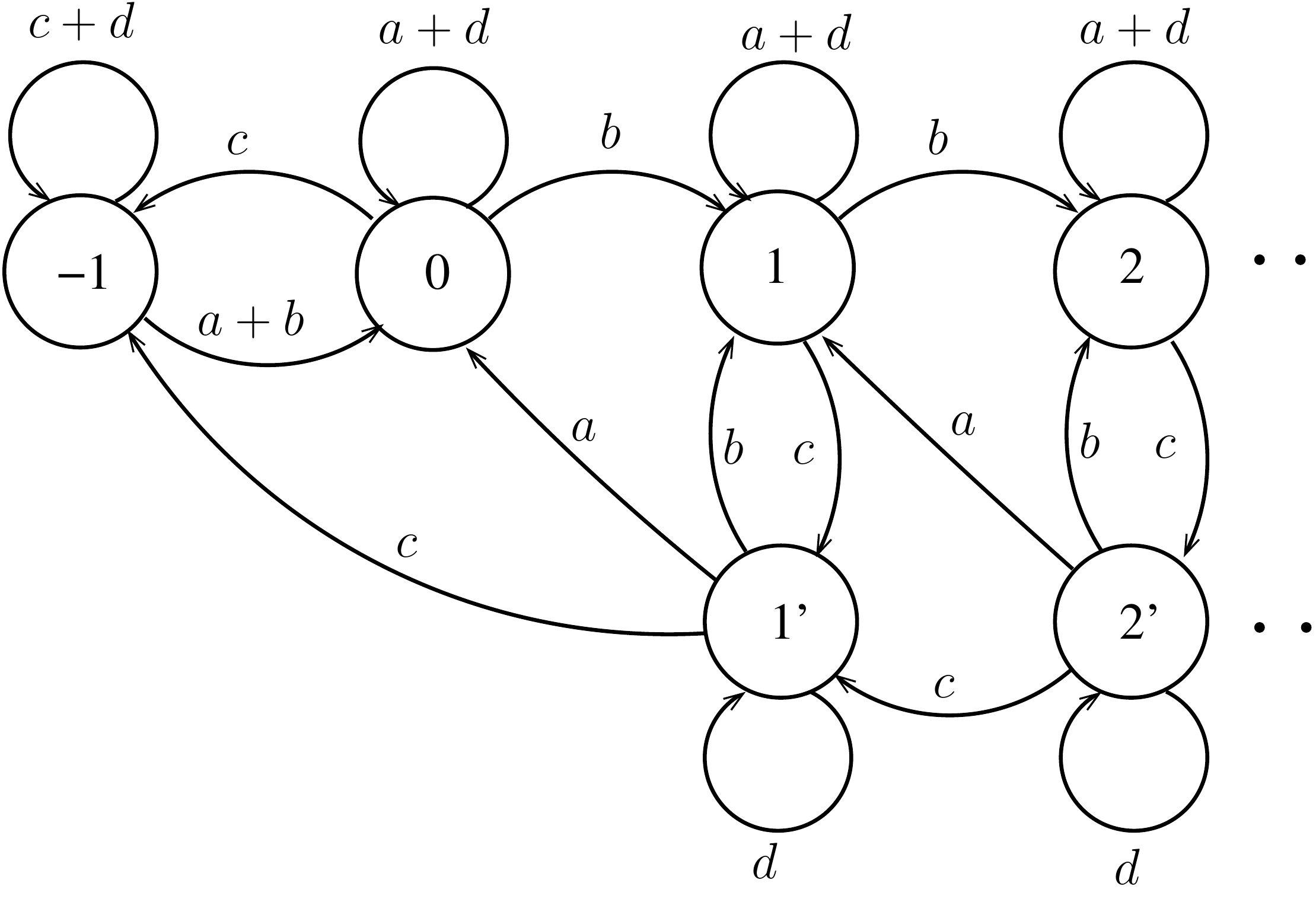}
\caption{Markov chain model of packet decoding with the coding scheme given by Claim~\ref{clm:fixed_prim}, where $U_1$ is the primary user. The state index $i$ represents the number of gaps in decoding of $U_2$ minus that for $U_1$. The states $i'$ are the advantage states where $U_2$ gets a chance to decode its required packet.\label{fig:fixed_primary_markov}}
\vspace{-0.6cm}
\end{figure}
By Claim~\ref{clm:fixed_prim}, the source transmits $s_{r_1} \xor s_{r_2}$ when the system is in an advantage state $i'$, and it transmits $s_{r_1}$ when the system is in state $i$ for $i \geq -1$. We now describe the state transitions of this Markov chain. First observe that with probability $d= (1-p_1)(1-p_2)$, both users experience erasures and the system transitions from any state to itself. When the system is in state $-1$, the source transmits $s_{r_1}$. Since $s_{r_1}$ has been already decoded by $U_2$, the probability $c=p_2(1-p_1)$ erasure also keeps the system in the same state. If the channel is successful for $U_1$, which occurs with probability $p_1 = a+b$, it fills its decoding gap and the system goes to state $0$.

The source transmits $s_{r_1}$ in any state $i$, $i \geq 1$. With probability $a = p_1 p_2$, both users decode $s_{r_1}$, and hence the state index $i$ remains the same. With probability $b = p_1(1-p_2)$, $U_1$ receives $s_{r_1}$ but $U_2$ does not, causing a transition to state $i+1$. With probability $c=(1-p_1) p_2$, $U_2$ receives $s_{r_1}$ and $U_1$ experiences an erasure due to which the system moves to the advantage state $i'$. When the system is an advantage state, having decoded $s_{r_1}$ gives $U_2$ an advantage because it can use $s_{r_1} \xor s_{r_2}$ transmitted in the next slot to decode $s_{r_2}$. From state $i'$, with probability $a$, $U_1$ decodes $s_{r_1}$ and $U_2$ decodes $s_{r_2}$, and the state transitions to $i-1$. With probability $c$, $U_2$ decodes $s_{r_2}$, but $U_1$ does not decode $s_{r_1}$. Thus, the system goes to state $(i-1)'$, except when $i=1$, where it goes to state $0$.

\begin{clm}
\label{clm:stability_fixed_primary_markov}
The Markov chain in Fig.~\ref{fig:fixed_primary_markov} is positive-recurrent and has unique steady-state distribution if and only if $b<c$, which is equivalent to $p_1 < p_2$.
\end{clm}

\begin{lem}[Trade-off for the Piggybacking user]
\label{lem:piggybacking_user_trade-off}
When the source always gives priority to user $U_1$ it achieves the optimal trade-off $(\tau_1, \lambda_1) = (p_1,-\log(1-p_1))$. The scheme in Claim~\ref{clm:fixed_prim} gives the best achievable $(\tau_2, \lambda_2)$ trade-off for piggybacking user $U_2$. The throughput $\tau_2$ is given by
\begin{align}
\tau_2 &= p_1 \quad \text{if  } p_2 > p_1.\label{eqn:tau_fixed_prim}
\end{align}
If $p_2 \leq p_1$, $\tau_2$ cannot be evaluated using our Markov chain analysis. The inter-delivery exponent of $U_2$ for any $p_1$ and $p_2$ is given by
\begin{align}
 \lambda_2 &= -\log\left( \max \left(\frac{\left( 1-c+d + \sqrt{ (1-c+d)^2 + 4(bc+cd-d)} \right)}{2}, 1-p_1\right) \right)\label{eqn:lambda_fixed_prim}.
\end{align}
\end{lem}

\begin{figure}[t]
\centering
\includegraphics[width=3.5in]{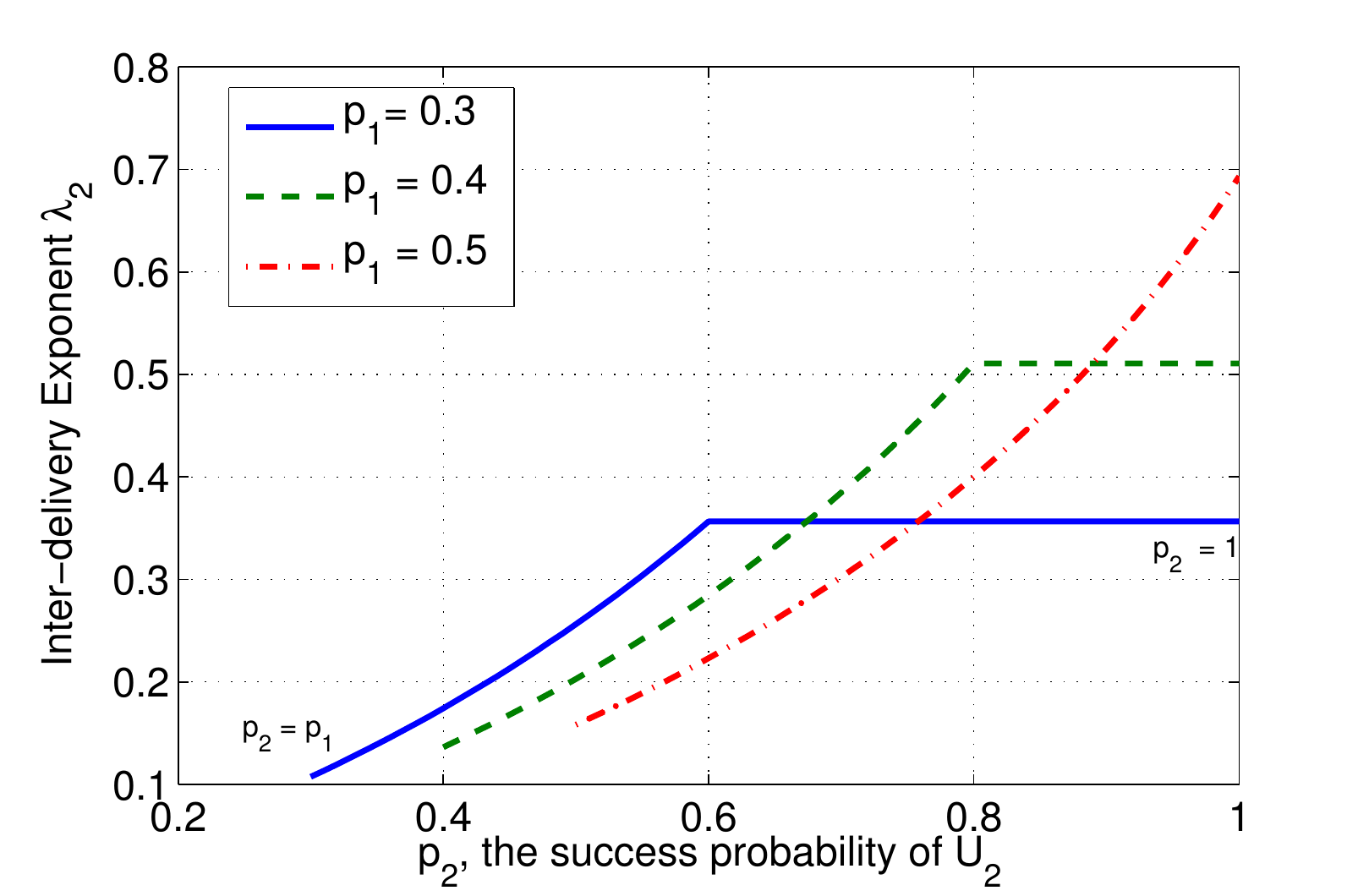}
\caption{Plots of the inter-delivery exponent $\lambda_2$ of the piggybacking user $U_2$, versus the success probability $p_2$ throughput $\tau_2$. The value of $p_2$ varies from $p_1$ to $1$ on each curve. The exponent saturates at $-\log(1-p_1)$, which is equal to $\lambda_1$, the exponent of the primary user $U_1$. \label{fig:thpt_in_order_fixed_markov}}
\end{figure}

In Fig.~\ref{fig:thpt_in_order_fixed_markov} we plot the inter-delivery exponent $\lambda_2$ versus $p_2$, which increases from $p_1$ to $1$ along each curve. The inter-delivery exponent $\lambda_2$ increases with $p_2$, but saturates at $- \log(1-p_1)$, the inter-delivery exponent $\lambda_1$ of the primary user. Since $U_2$ is the secondary user it cannot achieve faster in-order delivery than the primary user $U_1$.

%

%
\subsection{General Throughput-Smoothness Trade-offs}
\label{subsec:general_trade-off}
For the general case, we propose coding schemes that can be combined to tune the priority given to each of the two users and achieve different points on their throughput-smoothness trade-offs. 

Let $r_{\text{max}} = \max(r_1, r_2)$ and $r_{\text{min}} = \min (r_1 ,r_2)$, where $r_1$ and $r_2$ are the indices of the required packets of the two users. We refer to the user with the higher index $r_i$ as the leader(s) and the other user as the lagger. Thus, $U_1$ is the leader and $U_2$ is the lagger when $r_1 > r_2$. If $r_1 = r_2$, without loss of generality we consider $U_1$ as the leader. 

\begin{defn}[Priority-$(q_1, q_2)$ Codes]
\label{defn:priority_q1_q2}
If the lagger $U_i$ has not decoded packet $s_{r_{\text{max}}}$, the source transmits $s_{r_{\text{min}}}$ with probability $q_i$ and $s_{r_{\text{max}}}$ otherwise. If the lagger has decoded $s_{r_{\text{max}}}$ , the source transmits $s_{r_{\text{max}}} \xor s_{r_{\text{min}}}$.
\end{defn}

Note that the code given in Claim~\ref{clm:fixed_prim}, where the source always gives priority to user $U_1$ is a special case of priority-$(q_1, q_2)$ codes with $(q_1, q_2) = (1,0)$. Another special case is $(q_1, q_2) = (0,0)$ which is a greedy coding scheme that always favors the user which is ahead in in-order delivery. The greedy coding scheme ensures throughput optimality to both users, i.e.\ $\tau_1 = p_1$ and $\tau_2 = p_2$. 

\begin{rem}
A generalization of priority-$(q_1, q_2)$ codes is to consider priorities $q_1^{(i)}$ and $q_2^{(i)}$ that depend on the state $i$ of the Markov chain. A special case of this is $q_1^{(i)} = 1$ for all states $i \geq -M$, and $q_2^{(i)} = 1$ for all states $j \leq N$ for integers $M, N > 0$. This scheme corresponds to putting hard boundaries on both sides of the Markov chain, and was analyzed in \cite{gauri_netcod}.
\end{rem}

\begin{figure}
\centering
\includegraphics[width=4.5in]{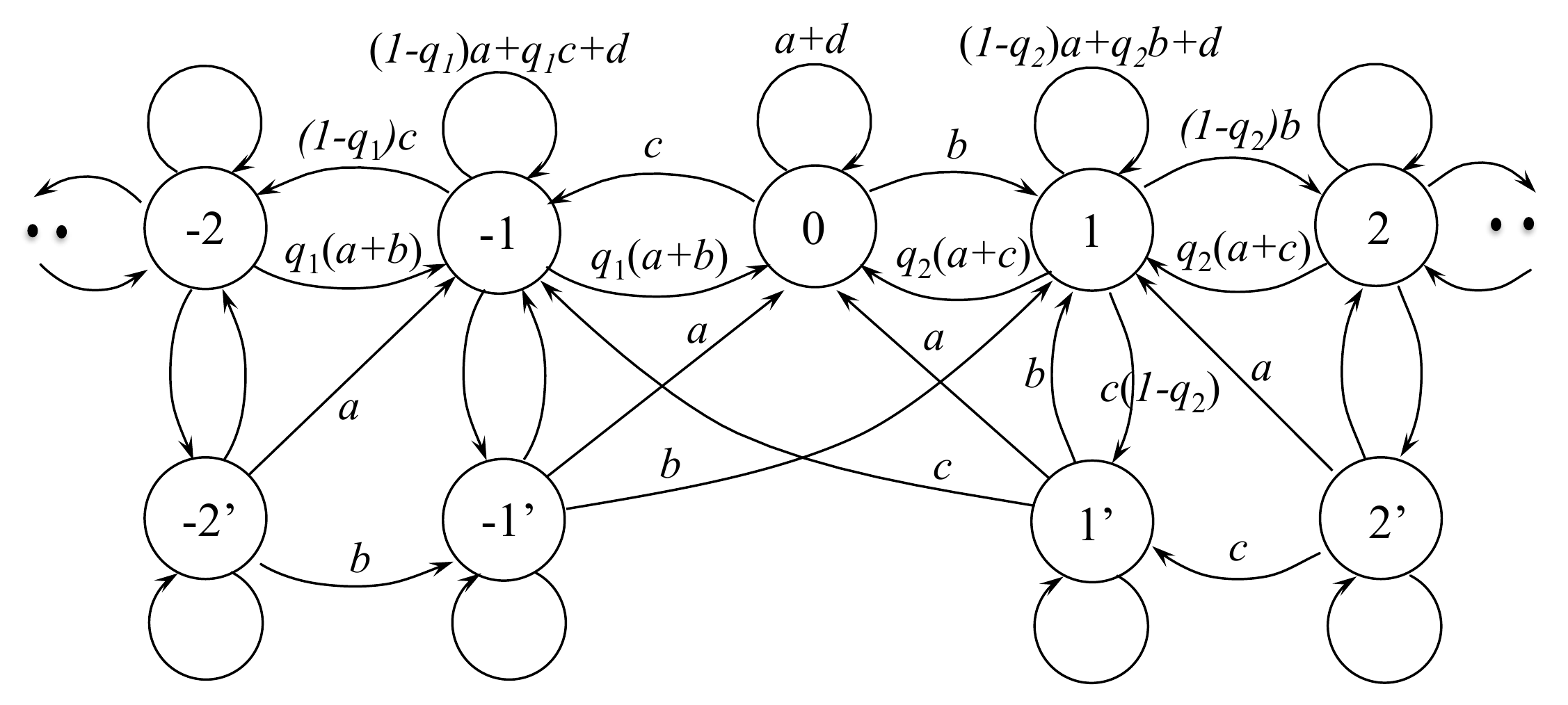}
\caption{ Markov chain model of packet decoding with the priority-$(q_1, q_2)$ coding scheme given by Definition~\ref{defn:priority_q1_q2}. The state index $i$ represents the number of gaps in decoding if $U_2$ compared to $U_1$ and $q_i$ is the probability of giving priority to the $U_i$ when it is the lagger, by transmitting its required packet $s_{r_i}$. and \label{fig:stabilized_markov}}
\end{figure}

The Markov model of packet decoding with a priority-$(q_1, q_2)$ code is as shown in Fig.~\ref{fig:stabilized_markov}, which is a two-sided version of the Markov chain in Fig.~\ref{fig:fixed_primary_markov}. Same as in Fig.~\ref{fig:fixed_primary_markov}, the index $i$ of a state $i$ of the Markov chain is the number of gaps in decoding of $U_2$ minus that for $U_1$. User $U_1$ is the leader when the system is in state $i \geq 1$ and $U_2$ is the leader when $ i \leq -1$, and both users are leaders when $i =0$. The system is in the advantage state $i'$ if packet is decoded by the lagger but not the leader. 

For simplicity of representation we define the notation $\bar{d} \triangleq 1-d$, $\bar{q_1} \triangleq 1-q_1$ and $\bar{q_2} \triangleq 1-q_2$. 
%

\begin{lem}[$(\tau, \lambda)$ Trade-offs with Priority-$(q_1, q_2)$ codes]
\label{lem:priority_q1_q2_trade-off}
Let $\mu = \pi_{i-1}/\pi_i$ for $i \leq -1$. Then the priority-$(q_1, q_2)$ codes given by Definition~\ref{defn:priority_q1_q2} give the following throughput for $U_2$.
\begin{align}
\tau_2 &= p_2 \left(1 - \frac{q_1 \pi_{-1}}{1-\mu}\right) \quad \text{if }{\mu < 1}  \label{eqn:tau_2_general}
\end{align}
 If $\mu > 1$ then $\tau_2$ cannot be evaluated using our Markov chain analysis. On the other hand, the inter-delivery exponent can be evaluated for any $\mu$ as given by
\begin{align}
\lambda_2 &= -\log \max \left( d+q_1 c + \bar{q_1} b, \frac{\left( 2d+\bar{q_2}a + b+ \sqrt{(2d+\bar{q_2}a +b)^2 - 4( d(b+d) +  \bar{q_2}(da-bc))} \right)}{2} \right)  \label{eqn:lambda_2_general}
\end{align}
Similarly, let $\rho = \pi_{i+1}/\pi_i$ for $ i \geq 1$. If $\rho < 1$, the expressions for throughput $\tau_1$ and inter-delivery exponent $\lambda_1$ of $U_1$ are same as \eqref{eqn:tau_2_general} and \eqref{eqn:lambda_2_general} with $b$ and $c$, and $q_1$ and $q_2$ interchanged, $\pi_{-1}$ replaced by $\pi_{1}$, and $\mu$ replaced by $\rho$.
\end{lem}



%

Fig.~\ref{fig:thpt_smoothness_q2_var} shows the throughput-smoothness trade-offs of the two users as $q_2$ varies, when $q_1 = 1$, $p_1 = 0.5$ and $p_2 = 0.4$. To stabilize the right-side of the Markov chain in Fig.~\ref{fig:stabilized_markov} for these parameters we require $q_2$ to be at least $0.25$. As $q_2$ increases from $0.26$ to $1$ in Fig.~\ref{fig:thpt_smoothness_q2_var} we observe that $U_2$ gains in smoothness, at the cost of the smoothness of $U_1$. Also, both users lose throughput when $q_2$ increases.

In Fig.~\ref{fig:p_1_eq_p2_q_1_eq_q2_var} we show the effect of increasing $q_1$ and $q_2$ simultaneously for different values of $p_1 = p_2$. As $q_1 = q_2$ increases, we get a better inter-delivery exponent for both users, but at the cost of loss of throughput. 
%

\begin{figure}[t]
\centering
\includegraphics[width=3.3in]{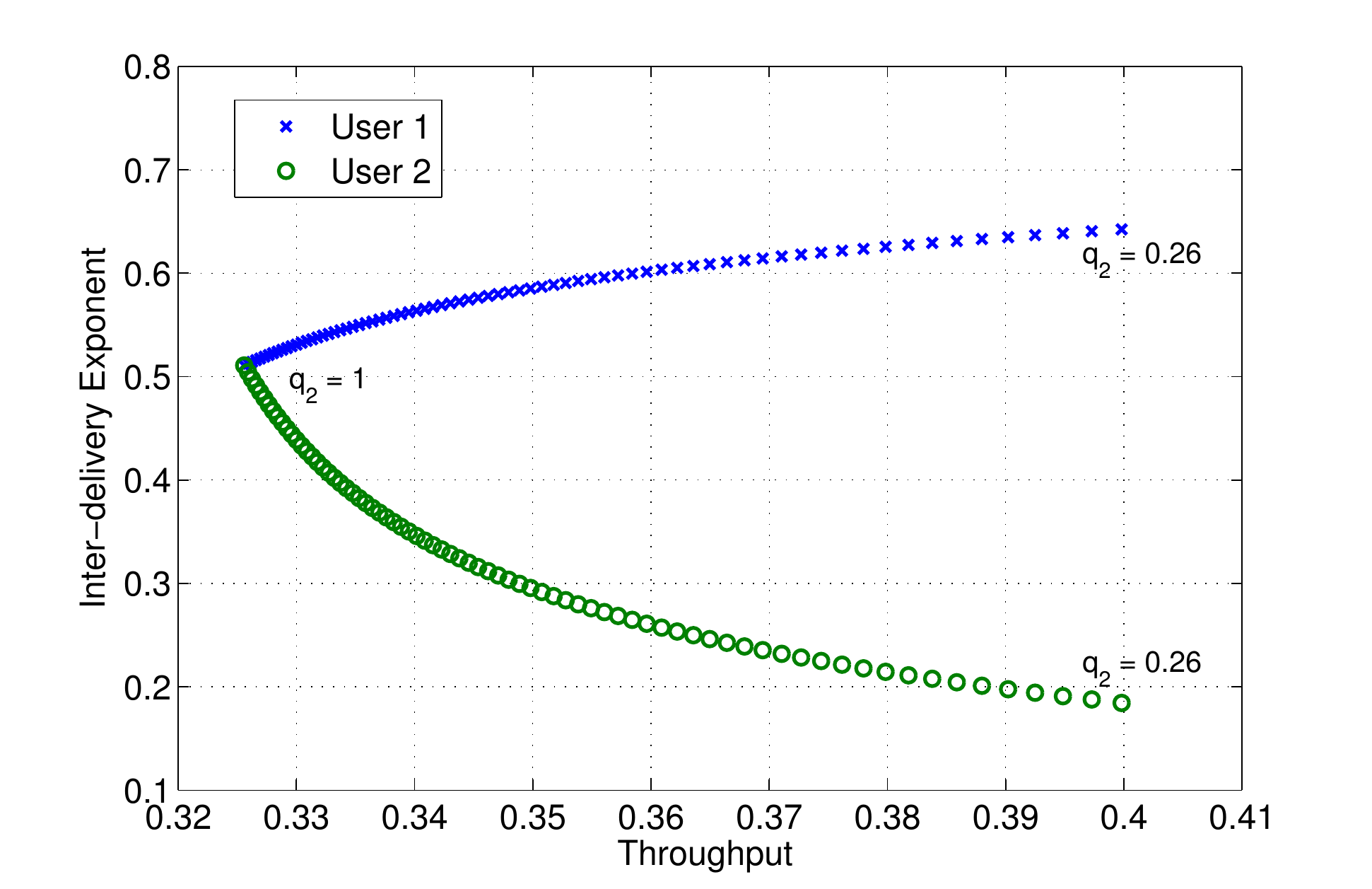}
\caption{Plot of the throughput-smoothness trade-off for $q_1 = 1$ and as $q_2$ varies. The success probabilities $p_1 = 0.5$ and $p_2 = 0.4$.\label{fig:thpt_smoothness_q2_var}}
\end{figure}

\begin{figure}[t]
\centering
\includegraphics[width=3.3in]{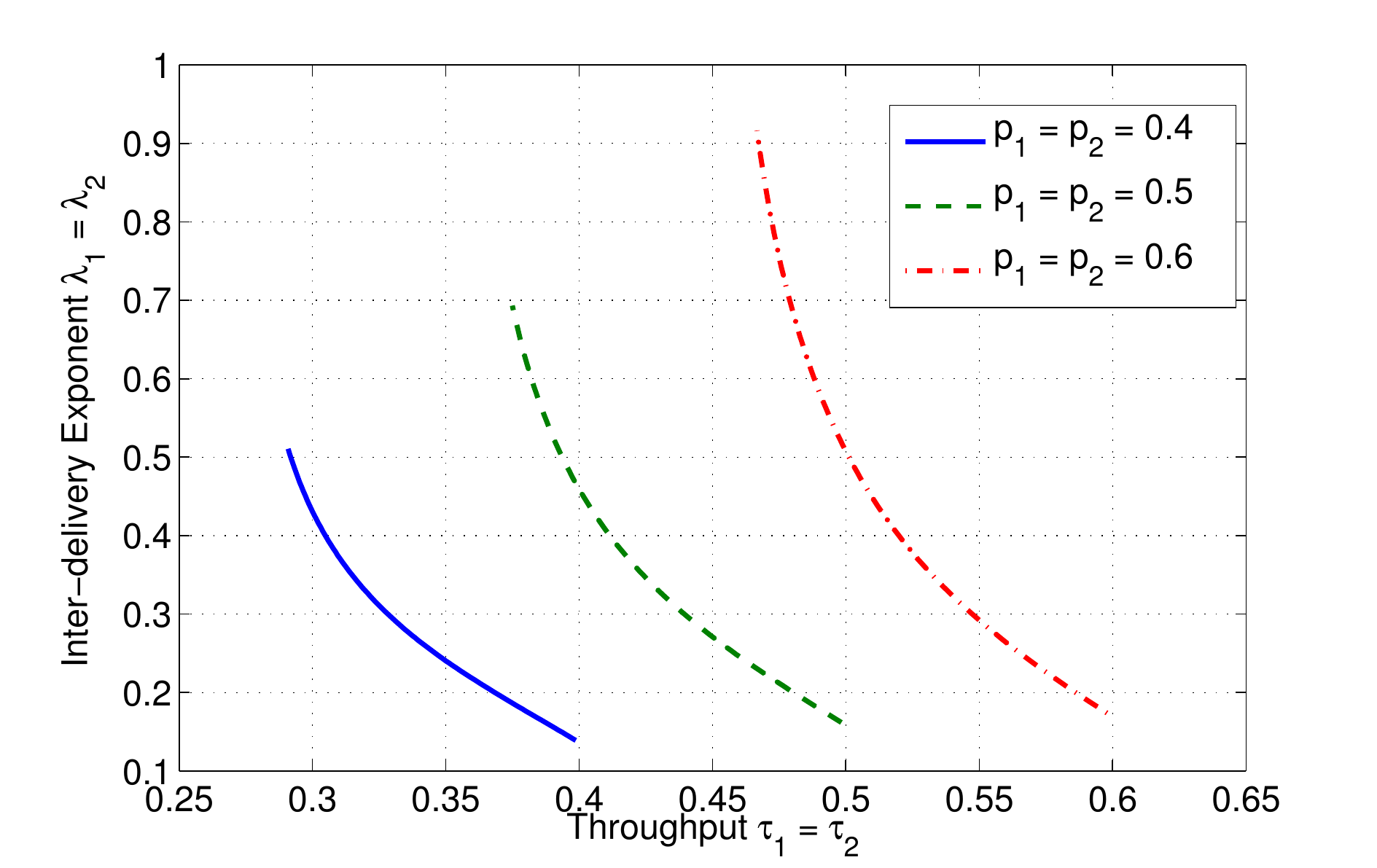}
\caption{Plot of the throughput-smoothness trade-off for different values $p_1 = p_2$. On each curve, $q_1 = q_2$ varies from $0$ from $1$.\label{fig:p_1_eq_p2_q_1_eq_q2_var}}
\end{figure}



\section{Concluding Remarks}
\label{sec:conclu}
\pdfoutput=1
\subsection{Major Implications}

In this paper we consider the problem of streaming over an erasure channel when the packets are required in-order by the receiver application. We investigate the trade-off between the throughput and the smoothness of in-order packet delivery. The smoothness is measured in terms of the smoothness exponent, which we show to be equivalent to an easy-to-analyze metric called the inter-delivery exponent. 

We first study the effect of block-wise feedback affects the throughput and smoothness of packet delivery in point-to-point streaming.  Our analysis shows that frequent feedback drastically improves the smoothness, for the same throughput. We present a spectrum of coding schemes that span different points on the throughput-smoothness trade-off. Depending upon the delay-sensitivity and bandwidth limitations of the applications, one can choose a suitable operating point on this trade-off.

Next we consider the problem where multiple users are streaming data from the source over a broadcast channel, over independent erasure channels with immediate feedback. Since different users decode different sets of packets, the source has to strike balance between giving priority to ensuring in-order packet decoding at each of the users. We study the inter-dependence between the throughput-delay trade-offs for the case of two users and develop coding schemes to tune the priority given to each user. 



\subsection{Future Perspectives}

In this paper, we assume strict in-order packet delivery, without allowing packet dropping, which can improve the smoothness exponent. Determining how significant this improvement is remains to be explored. Another possible future research direction is to extend the multicast streaming analysis to more users using the Markov chain-based techniques developed here for the two user case. A broader research direction is to consider the problem of streaming from distributed sources that are storing overlapping sets of the data. Using this diversity can improve the delay in decoding each individual packet. 

\begin{appendix}
\pdfoutput=1
\begin{proof}[Proof of \Cref{thm:equiv_exponents}]
The in-order decoding delay $D_k$ of packet $s_k$ can be expressed as a sum of inter-delivery times as follows.
\begin{align}
D_k  = T_1 + T_2 + \dots T_W 
\end{align}
where $W$ is the number of in-order delivery instants until packets $s_1$, \dots $s_k$ are decoded. The random variable $W$ can take values $1 \leq W \leq k$, since multiple packets may be decoded at one in-order delivery instants. Note that successive inter-delivery times $T_1, T_2, \dots, T_W$ are not i.i.d. The tail probability $\Pr(T_1 > t)$ of the first inter-delivery time is,
\begin{align}
\Pr(T_1 > t) &\geq  \Pr(T_i > t) \text{ for all integers } i, t \geq 0. \label{eqn:T_1_slowest}
\end{align} 
This is because during the first inter-delivery time $T_1$ we start with no prior information. During time $T_1$, the receiver may collect coded combinations that it is not able to decode. For a time-invariant scheme, these coded combinations can result in faster in-order decoding and hence a smaller $T_i$ for $ i > 1$. 


We now find lower and upper bounds on $\Pr(D_k \geq n)$ to find the decay rate of $D_k$. The lower bound can be derived as follows.
\begin{align}
\Pr( D_k \geq n) &= \mathbb{E}_{W} \left[\Pr(T_1 + T_2 + .. T_W \geq n) \right]\\
&\geq \Pr(T_1 \geq n) \label{eqn:D_k_lower_bnd_1} \\
&\doteq e^{-\lambda n} \label{eqn:D_k_lower_bnd_2}
\end{align}
where \eqref{eqn:D_k_lower_bnd_2} follows from Definition~\ref{defn:inter_deli_exponent}. Now we derive an upper bound on $\Pr(D_k > n)$.

\begin{align}
\Pr(D_k \geq n) &=\mathbb{E}_{W}  \left[ \Pr(T_1 + T_2 + \dots + T_W \geq n) \right] \\
&\leq \mathbb{E}_{W}  \left[ \Pr(T_1^{(1)} + T_1^{(2)} + \dots + T_1^{(W)}\geq n)  \right]\\
&= \mathbb{E}_{W}  \left[ \sum_{n_i: \sum_{i=1}^{W} n_i = n} \prod_{i=1}^{W} \Pr(T_1^{(i)} > n_i)  \right] \label{eqn:D_k_upper_bnd_3} \\
&\doteq \mathbb{E}_{W}  \left[ \sum_{n_i: \sum_{i=1}^{W} n_i = n} e^{-\lambda (n_1 + \dots + n_W)}  \right] \label{eqn:D_k_upper_bnd_4}\\
&= \mathbb{E}_{W}  \left[ \binom{n +W-1}{ W-1}e^{-\lambda n}  \right] \label{eqn:D_k_upper_bnd_5}\\
&\doteq  \mathbb{E}_{W}  \left[ e^{-\lambda n}  \right] \label{eqn:D_k_upper_bnd_6}\\
&= e^{-\lambda n} \label{eqn:D_k_upper_bnd_7}
\end{align}
where $T_1^{(i)}$ are i.i.d.\ samples from the probability distribution of the first inter-delivery time $T_1$. By \eqref{eqn:T_1_slowest}, replacing $T_i$ by $T_1^{(i)}$ gives an upper bound on the probability $\Pr(D_k \geq n)$. Since $T_1^{(i)}$ are i.i.d.\ we can express the upper bound as a product of tail probabilities in \eqref{eqn:D_k_upper_bnd_3}, where $n_i$ are non-negative integers summing to $n$. By \eqref{eqn:decay_rate_T1}, each term in the product in \eqref{eqn:D_k_upper_bnd_3} asymptotically decays at rate $\lambda$. Thus we get \eqref{eqn:D_k_upper_bnd_4} and \eqref{eqn:D_k_upper_bnd_5}. Since $W \leq k << n$, the binomial coefficient decays subexponentially, and we get \eqref{eqn:D_k_upper_bnd_7}. 

From \eqref{eqn:D_k_lower_bnd_2} and \eqref{eqn:D_k_upper_bnd_7} we can conclude that the asymptotic decay rate $\lambda_k^{(s)}$ for any $k$ is equal to the inter-delivery exponent $\lambda$.
\end{proof}

\vspace{0.5cm}
\begin{proof}[Proof of \Cref{thm:opt_no_feedback}]
 We first show that the scheme with transmit index $V[n] = \lceil rn \rceil$ in time slot $n$ achieves the trade-off $(\tau, \lambda) = (r , D(r \| p))$. Then we prove the converse by showing that no other full-rank scheme gives a better trade-off.

\textbf{Achievability Proof:} Consider the scheme with transmit index $V[n] = \lceil r n\rceil$, where $r$ represents the rate of adding new packets to the transmitted stream. The rate of adding packets is below the capacity of the erasure channel. Thus it is easy to see that the throughput $\tau = r$. Let $E[n]$ be the number of combinations, or equations received until time $n$. It follows the binomial distribution with parameter $p$. All packets  $s_1 \cdots s_{V[n]}$ are decoded when $E[n] \geq V[n]$. Define the event $G_n = \{ E[j] < V[j] \text{ for all } 1 \leq j \leq n\}$, that there is no packet decoding until slot $n$. The tail distribution of the first inter-delivery time $T_1$ is,
\begin{align}
\Pr(T_1>n) &= \sum_{k=0}^{\lceil n r \rceil - 1}  \Pr(E[n]=k) \Pr(G_n | E[n] = k),\nonumber\\
&= \sum_{k=0}^{\lceil n r \rceil - 1}  \binom{n}{k} p^{k} (1-p)^{n-k} \Pr(G_n | E[n] = k), \nonumber
\end{align}
where $\Pr(G_n | E[n] = k) = 1-k/n$ as given by the Generalized Ballot theorem in \cite[Chapter~4]{durrett}. Hence it is sub-exponential and does not affect the exponent of $\Pr(T_1> n)$ and we have
\begin{align}
\Pr(T_1>n) & \doteq \sum_{k=0}^{\lceil n r \rceil - 1}  \binom{n}{k} p^{k} (1-p)^{n-k},  \label{eqn:asym_T_no_fb_1}\\
& \doteq \binom{n}{\lceil n r \rceil - 1} p^{\lceil n r \rceil - 1} (1-p)^{n-\lceil n r \rceil + 1},  \label{eqn:asym_T_no_fb_2}\\
& \doteq e^{-n D(r \| p)}, \label{eqn:asym_T_no_fb_3}
\end{align}
where in \eqref{eqn:asym_T_no_fb_1} we take the asymptotic equality $\doteq$ to find the exponent of $\Pr(T_1>n)$, and remove the $\Pr(G_n | E[n] = k)$ term because it is sub-exponential. In \eqref{eqn:asym_T_no_fb_2}, we only retain the $k = \lceil n r \rceil - 1$ term from the summation because for $r \leq p$, that term asymptotically dominates other terms. Finally, we use the Stirlings approximation $\binom{n}{k} \approx e^{n H(k/n)}$ to obtain \eqref{eqn:asym_T_no_fb_3}.


\vspace{0.2cm}
\textbf{Converse Proof:}
First we show that the transmit index $V[n]$ of the optimal full-rank scheme should be non-decreasing in $n$. Given any scheme, we can permute the order of transmitting the coded packets such that $V[n]$ is non-decreasing in $n$. This does not affect the throughput $\tau$, but it can improve the inter-delivery exponent $\lambda$ because decoding can occur sooner when the initial coded packets include fewer source packets.

We now show that it is optimal to have $V[n] = \lceil rn \rceil$, where we add new packets to the transmitted stream at a constant rate $r$. Suppose a full-rank scheme uses rate $r_i$ for $n_i$ slots for all $1 \leq i \leq L$, such that $\sum_{i=0}^{L} n_i = n$ and $\sum_{i=1}^{L} n_i r_i = nr$. Then, the tail distribution of $T_1$ is,
\begin{align}
\Pr(T_1>n) &= \sum_{k=0}^{\lceil \sum_{i=1}^{L} n_i r_i \rceil - 1}  \Pr(E[n]=k) \Pr(G_n | E[n] = k),\label{eqn:asym_T_no_fb_4}\\
&\doteq \sum_{k=0}^{\lceil n r \rceil - 1}  \binom{n}{k} p^{k} (1-p)^{n-k}, \label{eqn:asym_T_no_fb_5} \\
& \doteq e^{-n D(r \| p)}. \label{eqn:asym_T_no_fb_6}
\end{align}
Varying the rate of adding packets affects the term $\Pr(G_n | E[n] = k)$ in \eqref{eqn:asym_T_no_fb_4}, but it is still $\omega(1/n)$ and we can eliminate it when we take the asymptotic equality in \eqref{eqn:asym_T_no_fb_5}. As a result, the in-order delay exponent is same as that if we had a constant rate $r$ of adding new packets to the transmitted stream. Hence we have proved that no other full-rank scheme can achieve a better $(\tau, \lambda)$ trade-off than $V[n] = \lceil nr \rceil$ for all $n$.
\end{proof}

\begin{proof}[Proof of \Cref{lem:interpolate_bw_schemes}]
Here we prove the result for $B=2$, that is randomizing between two schemes. It can be extended to general $B$ using induction. Given two time-invariant schemes $\mathbf{x}^{(1)}$ and $\mathbf{x}^{(2)}$ that achieve the throughput-delay trade-offs $( \tau_{\mathbf{x}^{(1)}}, \lambda_{\mathbf{x}^{(1)}})$ and $(\tau_{\mathbf{x}^{(2)}},\lambda_{\mathbf{x}^{(2)}})$ respectively, consider a randomized strategy where, in each block we use the scheme $\mathbf{x}^{(1)}$ with probability $\mu$ and scheme $\mathbf{x}^{(2)}$ otherwise. Then, it is easy to see that the throughput on the new scheme is $\tau = \mu \tau_{\mathbf{x}^{(1)}} + (1-\mu) \tau_{\mathbf{x}^{(2)}}$. 

Now we prove the inter-delivery exponent $\lambda$ is also a convex combinations of $\lambda_{\mathbf{x}^{(1)}}$ and $\lambda_{\mathbf{x}^{(2)}}$. Let $p_{d_1}$ and $p_{d_2}$ be the probabilities of decoding the first unseen packet in a block using scheme $\mathbf{x}^{(1)}$ and $\mathbf{x}^{(2)}$ respectively. Suppose in an interval with $k$ blocks, we use scheme $\mathbf{x}^{(1)}$ for $h$ blocks, and scheme $\mathbf{x}^{(2)}$ in the remaining blocks, we have 
\begin{equation}
\Pr(T_1 > kd) = (1-p_{d_1})^{h} (1-p_{d_2})^{k-h}.
\end{equation}
Using this we can evaluate $\lambda$ as,
\begin{align}
\lambda &= \lambda_{\mathbf{x}^{(1)}} \lim_{k \rightarrow \infty} \frac{h}{k} +\lambda_{\mathbf{x}^{(2)}} \lim_{k \rightarrow \infty} \frac{k-h}{k} \label{eqn:lambda_convex_comb} \\
 &= \mu \lambda_{\mathbf{x}^{(1)}} + (1-\mu) \lambda_{\mathbf{x}^{(2)}}
\end{align}
where we get \eqref{eqn:lambda_convex_comb} using \eqref{eqn:lambda_d}. As $k \rightarrow \infty$, by the weak law of large numbers, the fraction $h/k$ converges to $\mu$. 
\end{proof}

\begin{proof}[Proof of Lemma~\ref{lem:d_2_3_tradeoff}]
When $d=2$ there are only two possible time-invariant schemes $\mathbf{x} = [2,0]$ and $[1, 1]$ that give unique $(\tau, \lambda)$. By Remark~\ref{rem:uniqueness_of_schemes}, all other $\mathbf{x}$ are equivalent to one of these vectors in terms $(\tau, \lambda)$. The vectors $\mathbf{x} = [2,0]$ and $[1, 1]$ correspond to the $a=1$ and $a=2$ codes proposed in Definition~\ref{defn:gen_d_codes}. Hence, the line joining their corresponding $(\tau, \lambda)$ points, as shown in Fig.~\ref{fig:close_to_optimal_tradeoff}, is the best trade-off for $d=2$. 

When $d=3$ there are four time-invariant schemes $\mathbf{x}^{(1)} = [1, 0, 2]$, $\mathbf{x}^{(2)} = [2,1,0]$, $\mathbf{x}^{(3)} =[1,2,0]$ and $\mathbf{x}^{(4)} = [3,0,0]$ that give unique $(\tau, \lambda)$, according to Definition~\ref{defn:time_invariant} and Remark~\ref{rem:uniqueness_of_schemes}. The vectors $\mathbf{x}^{(1)}$, $\mathbf{x}^{(2)}$ and $\mathbf{x}^{(4)}$ correspond to the codes with $a = 1,2,3$ in Definition~\ref{defn:gen_d_codes}. The throughput-delay trade-offs $(\tau_{\mathbf{x}^{(i)}},\lambda_{\mathbf{x}^{(i)}})$ for $i = 1,2,4$ achieved by these schemes are given by \eqref{eqn:close_to_optimal}. From Claim~\ref{clm:cost_of_opt_lambda} and Claim~\ref{clm:cost_of_opt_tau} we know that $(\tau_{\mathbf{x}^{(1)}},\lambda_{\mathbf{x}^{(1)}})$ and $(\tau_{\mathbf{x}^{(4)}},\lambda_{\mathbf{x}^{(4)}})$ have to be on the optimal trade-off. By comparing the slopes of the lines joining these points we can show that the point $(\tau_{\mathbf{x}^{(2)}},\lambda_{\mathbf{x}^{(2)}})$ lies above the line joining $(\tau_{\mathbf{x}^{(1)}},\lambda_{\mathbf{x}^{(1)}})$ and $(\tau_{\mathbf{x}^{(4)}},\lambda_{\mathbf{x}^{(4)}})$ for all $p$. Fig.~\ref{fig:close_to_optimal_tradeoff} illustrates this for $p=0.6$.
%
For the scheme with $\mathbf{x}^{(3)} =[1,2,0]$, we have
\[ (\tau_{\mathbf{x}^{(3)}},\lambda_{\mathbf{x}^{(3)}}) =\left( (3p-p^3)/3, -(\log(1-p)^2(1+p))/3\right). \]
Again, by comparing the slopes of the lines joining $(\tau_{\mathbf{x}^{(i)}},\lambda_{\mathbf{x}^{(i)}})$ for $i = 1, \cdots 4$ we can show that for all $p$, $(\tau_{\mathbf{x}^{(3)}},\lambda_{\mathbf{x}^{(3)}})$ lies below the piecewise linear curve joining $(\tau_{\mathbf{x}^{(i)}},\lambda_{\mathbf{x}^{(i)}})$ for $i = 1,2,4$.
%
\end{proof}

\begin{proof}[Proof on \Cref{clm:stability_fixed_primary_markov}]
We now solve for the steady-state distribution of this Markov chain. Let $\pi_i$ and $\pi'_i$ be the steady-state probabilities of states $i$ for $i \geq -1$ and advantages states $i'$ for all $i \geq 0$ respectively. The steady-state transition equations are given by
\begin{align}
(1-a-d) \pi_i  &=b (\pi_{i-1} + \pi'_i) + a \pi'_{i+1} \quad \text{ for } i \geq 1 \label{eqn:steady_state_4},\\
(1-d) \pi'_i &= c( \pi_i + \pi'_{i+1}) \quad \quad \quad \quad \,\,\,\, \text{ for } i \geq 1 \label{eqn:steady_state_5}, \\
(1-c-d) \pi_{-1} &= c( \pi_0+ \pi'_{1}) \label{eqn:steady_state_1},\\
(1-a-d) \pi_0 &=a \pi'_1 + (a+b)\pi_{-1} \label{eqn:steady_state_2}.
\end{align}

By rearranging the terms in \eqref{eqn:steady_state_4}-\eqref{eqn:steady_state_2}, we get the following recurrence relation,
\begin{equation}
\pi_i =\frac{(1-a-d)}{c} \pi_{i-1} - \frac{b}{c}\pi_{i-2}  \quad \text{ for } i \geq 2.
\label{eqn:pi_recurrence}
\end{equation}
Solving the recurrence in \eqref{eqn:pi_recurrence} and simplifying \eqref{eqn:steady_state_4}-\eqref{eqn:steady_state_2} further, we can express $\pi_i$, $\pi'_i$ for $i \geq 2$ in terms of $\pi_1$ as follows,
\begin{align}
\frac{\pi_i}{\pi_{i-1}}&= \frac{b}{c},  \label{eqn:steady_state_recursion_1}\\
\frac{\pi'_i}{\pi_i} &=\frac{c}{a+c}. \label{eqn:steady_state_recursion_2}
\end{align}
From \eqref{eqn:steady_state_recursion_1} we see that the Markov chain is positive-recurrent and a unique steady-state distribution if and only if $b<c$, which is equivalent to $p_1 < p_2$. If $p_1 \geq p_2$, the expected recurrence time to state $0$, that is the time taken for $U_2$ to catch up with $U_1$ is infinity. When the Markov chain is positive recurrent, we can use \eqref{eqn:steady_state_recursion_1} and \eqref{eqn:steady_state_recursion_2} to solve for all the steady state probabilities. 
\end{proof}

\begin{proof}[Proof of \Cref{lem:piggybacking_user_trade-off}]
Since we always give priority to the primary user $U_1$, we have $(\tau_1, \lambda_1) = (p_1,-\log(1-p_1))$. When $p_1 < p_2$, we can express the throughput $\tau_2$ in terms of the steady state probabilities of the Markov chain in Fig~\ref{fig:fixed_primary_markov}. User $U_2$ experiences a throughput loss when it is in state $-1$ and the next slot is successful. Thus, when $p_2 > p_1$,
\begin{align}
\tau_2 &= p_2 (1-\pi_{-1}), \\
&= p_2 \left(1-\frac{c-b}{a+c} \right) = p_1. \label{eqn:tau_2_stable} 
\end{align}
If $p_1 \geq  p_2$, the system drifts infinitely to the right side. There is a non-zero probability that in-order decoding via advantage states is not able to catch up and fill all gaps in decoding of $U_2$. Thus, we cannot evaluate $\tau_2$ using this Markov chain analysis. 

To determine $\lambda_2$, first observe that $U_2$ decodes an in-order packet when the system is in state $0$ or states $i'$, for $i \geq 1$, and the next slot is successful. As given by Definition~\ref{defn:inter_deli_exponent}, the inter-delivery exponent $\lambda_2$ is the asymptotic decay rate of $\Pr(T_1>t)$, the probability that no in-order packet is decoded by $U_2$ for $t$ consecutive slots. To determine $\lambda$, we add an absorbing state $F$ to the Markov chain as shown in Fig.~\ref{fig:fixed_primary_absorb}, such that the system transitions to $F$ when an in-order packet is decoded by $U_2$. 

 \begin{figure}[t]
\centering
\includegraphics[width=2.3in]{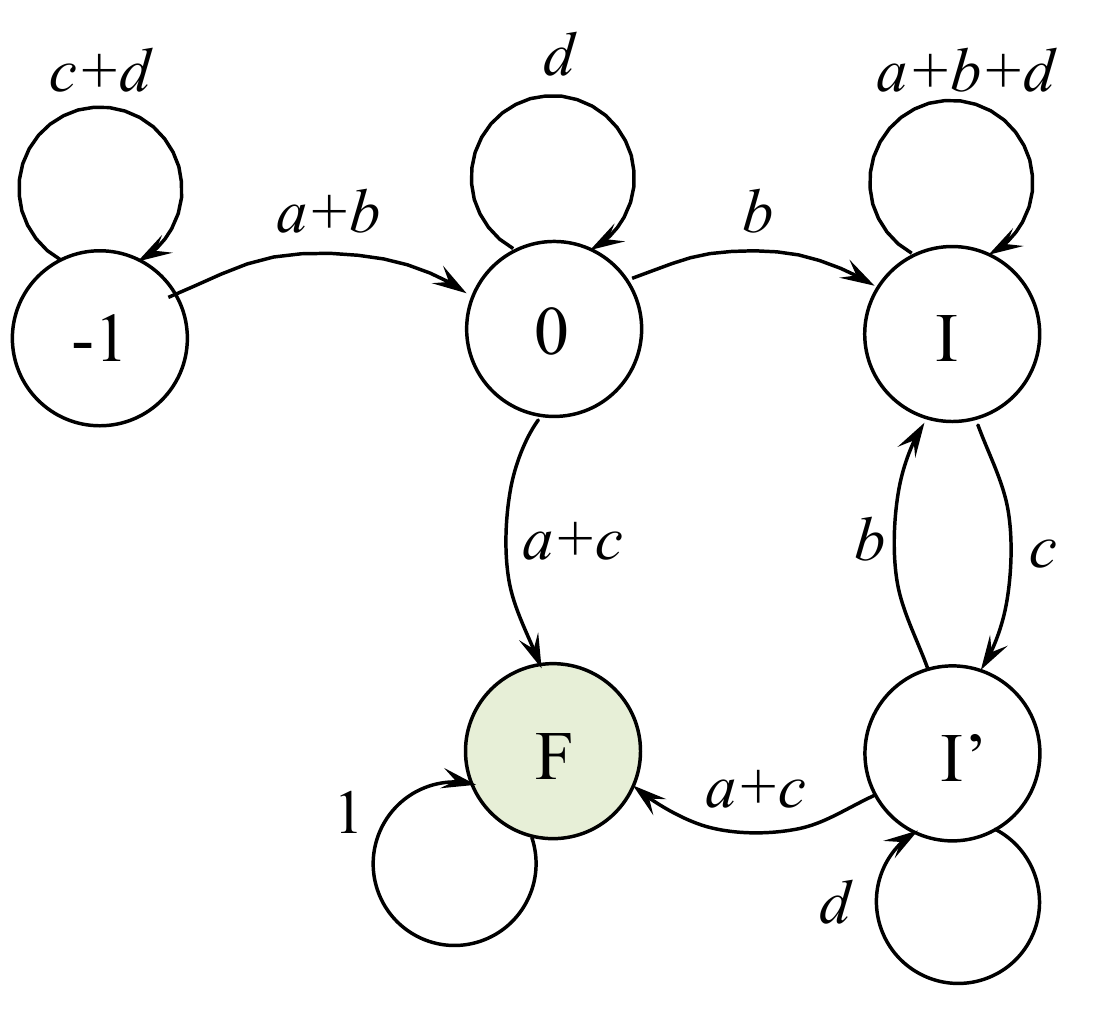}
\caption{Markov model used to determine the inter-delivery exponent $\lambda_2$ of user $U_2$. The absorbing state $F$ is reached when an in-order packet is decoded by $U_2$. The exponent of the distribution of the time taken to reach this state is $\lambda_2$. \label{fig:fixed_primary_absorb}}
\vspace{-0.6cm}
\end{figure}

In Fig.~\ref{fig:fixed_primary_absorb}, all the states $i$ and $i'$ for $i \geq 1$ are fused into states $I$ and $I'$ because this does not affect the probability distribution of the time to reach the absorbing state $F$. The inter-delivery exponent $\lambda_2$ is equal to the rate of convergence of this Markov chain to its steady state, which is known to be (see \cite[Chapter~4]{dsp_gallager})  $\lambda_2 = -\log \xi_2$ where $\xi_2$ is the second largest eigenvalue of the state transition matrix of the Markov chain,
\begin{align}
A &=  \left( \begin{array}{ccccc}
d & b & 0 & 0 & a+c \\
0 & a+b+d & c & 0 & 0 \\
0 & b & d & 0 & a+c \\
a+b & 0 & 0 & c+d & 0 \\
0 & 0 &  0 & 0 & 1 
\end{array} \right).
\end{align}

Solving for the second largest eigen-value of $A$, we can show that
\begin{align}
\xi_2 &= \max\left(1-p_1, \frac{\left( 1-c+d + \sqrt{ (1-c+d)^2 + 4(bc+cd-d)} \right)}{2} \right).
\end{align}
Hence the inter-delivery exponent $\lambda_2 = -\log \xi_2$ is as given by \eqref{eqn:lambda_fixed_prim}. 
\end{proof}

\begin{proof}[Proof of \Cref{lem:priority_q1_q2_trade-off}]
The state-transition equations of the Markov chain are as follows. 
\begin{align}
(\bar{d}-a) \pi_0 &= a (\pi'_1 + \pi'_{-1}) + q_1(a+b)  \pi_{-1}+ q_2(a+c)  \pi_1  \label{eqn:stab_steady_state_1}\\
(\bar{d}-\bar{q_2}a-q_2 b ) \pi_i  &= q_2 (a+c)  \pi_{i+1} +  \bar{q_2} b \pi_{i-1} + b \pi'_i + a\pi'_{i+1} \quad \quad \text{ for } i \geq 2 \label{eqn:stab_steady_state_2} \\
(\bar{d}-\bar{q_1} a -q_1 c ) \pi_i  &= q_1(a+b) \pi_{i+1} + \bar{q_1} c \pi_{i+1} + c \pi'_i + a \pi'_{i-1} \quad \quad \text{ for } i \leq -2 \label{eqn:stab_steady_state_3}\\
(\bar{d}-\bar{q_2}a -q_2 b ) \pi_1  &= q_2 (a+c)  \pi_{2} +  b( \pi'_{-1} + \pi_{0} + \pi'_1) + a\pi'_{2} \label{eqn:stab_steady_state_4}  \\
(\bar{d}-\bar{q_1}a-q_1 c ) \pi_{-1} &= q_1(a+b) \pi_{-2} + c( \pi'_{-1} + \pi_{0} + \pi'_1)+ a \pi'_{-2} \label{eqn:stab_steady_state_5}  \\
\bar{d} \pi'_i &= \bar{q_2}c \pi_i + c \pi'_{i+1} \quad\quad \text{ for } i \geq 1 \label{eqn:stab_steady_state_6}\\
\bar{d} \pi'_i &= \bar{q_2}c \pi_i + b \pi'_{i+1}  \quad\quad \text{ for } i \leq -1 \label{eqn:stab_steady_state_7}
\end{align}

Rearranging the terms, we get the following recurrence in the steady-state probabilities on the right-side of the chain,
\begin{align}
\alpha_3 \pi_{i+3}+ \alpha_2 \pi_{i+2} + \alpha_1 \pi_{i+1} +  \alpha_0 \pi_{i} &= 0  \quad \quad \text{ for } i \geq 1
\end{align}
where,
\begin{align}
\alpha_3 &= c(a+c) q_2 \\
 \alpha_2 &= -c \bar{d} + b c q_2 -(a+c) q_2 \bar{d} \\
 \alpha_1 &= \bar{d} (\bar{d} - b q_2 - a \bar{q_2}) \\
 \alpha_0 &= -\bar{d} b \bar{q_2} 
\end{align}

The characteristic equation of this recurrence has the roots $1$, $\rho$ and $\rho'$. We can show that both $\rho$ and $\rho'$ are positive and at least one of them is greater than $1$. The expression for the smaller root is,
\begin{align}
\rho &= -\frac{\alpha_3+ \alpha_2}{2 \alpha_3} - \frac{\sqrt{ (\alpha_3 +  \alpha_2)^2 + 4\alpha_3  \alpha_0}}{2 \alpha_3}
\end{align}

The right-side of the Markov chain is stable if and only if $\rho < 1$. Thus, when $\rho < 1$, the steady-state probabilities $\pi_i$ and $\pi'_i$ for $i \geq 1$ are related by the recurrences,
\begin{align}
\frac{\pi_{i+1}}{\pi_{i}} &= \rho \text{   and   } \frac{\pi'_{i}}{\pi_i} =\frac{ c(1-q_2)}{\bar{d}-c \rho} 
\end{align}

Similarly, for the left side of the chain we have the recurrences,
\begin{align}
\frac{\pi_{i+1}}{\pi_{i}} &= \mu  \text{   and   } \frac{\pi'_{i}}{\pi_i} =\frac{ b(1-q_1)}{\bar{d}-b\mu} 
\end{align}
where 
\begin{align}
\mu &=  -\frac{\beta_3+ \beta_2}{2 \beta_3} - \frac{\sqrt{ (\beta_3 +  \beta_2)^2 + 4\beta_3  \beta_0}}{2 \beta_3} 
\end{align}
with the expressions $\beta_k$ for $ k =0,1,2,3$ being the same as $\alpha_k$ with $b$ and $c$ interchanged, and $q_2$ replaced by $q_1$. We can use these recurrences we can express all steady-state probabilities $\pi_i$ and $\pi'_i$ for $i \geq 1$ in terms of $\pi_1$, and the steady-state probabilities $\pi_i$ and $\pi'_i$ for $i \leq  -1$ in terms of $\pi_{-1}$. Then using the states transition equation \eqref{eqn:stab_steady_state_1}, and the fact that all the steady state probabilities sum to $1$, we can solve for all the steady-state probabilities of the Markov chain.

User $U_1$ receives an innovative in every successful slot except when the source (with probability $q_2$) gives priority to $U_2$ in states $i$, $ i \geq 1$. Thus, if $\rho < 1$ its throughput is given by
\begin{align}
\tau_1 &= p_1 \left( 1- q_2 \sum_{i=1}^{\infty} \pi_i  \right)  = p_1 \left(1 - \frac{q_2 \pi_1}{1-\rho}\right)
\end{align}

Similarly if $\mu < 1$ we have,
\begin{align}
\tau_2 &= p_2 \left(1-q_1 \sum_{i=-\infty}^{-1} \pi_i \right) = p_2 \left(1 - \frac{q_1 \pi_{-1}}{1-\mu}\right)
\end{align} 

Similar to the proof of \Cref{lem:piggybacking_user_trade-off}, we determine the inter-delivery exponent $\lambda_2$ of user $U_2$ by adding an absorbing state $F$ to the Markov chain as shown in Fig.~\ref{fig:stab_markov_absorb_2}, such that the system transitions to $F$ when an in-order packet is decoded by $U_2$. In Fig.~\ref{fig:stab_markov_absorb_2}, all the states $i$ and $i'$ for $i \geq 1$ are fused into states $I$ and $I'$ because this does not affect the probability distribution of the time to reach the absorbing state $F$. The inter-delivery exponent $\lambda_2 = -\log \xi_2$ where $\xi_2$ is the second largest eigenvalue of its state transition matrix of this Markov chain which is given by, 
\begin{align}
A &=  \left( \begin{array}{ccccc}
d & b & 0 & 0 & a+c \\
0 & \bar{q_2} a+b+d & c \bar{q_2} & 0  & q_2(a+c) \\
0 & b & d & 0 & a+c \\
q_1(a+b) & 0 & 0 & d+ q_1 c+ \bar{q_1}b & \bar{q_1}(a+c) \\
0 & 0 & 0 & 0  & 1 
\end{array} \right).
\end{align}
Solving for the second largest eigen-value $\xi_2$ of $A$, we get
\begin{align}
\xi_2 &= \max \left(d+q_1 c + \bar{q_1} b,\frac{\left( 2d+\bar{q_2}a + b+ \sqrt{ (2d+\bar{q_2}a +b)^2 - 4( d(b+d) +  \bar{q_2}(da-bc))} \right)}{2} \right) .
\end{align}
The inter-delivery exponent $\lambda_2 = -\log \xi_2$ and is given by \eqref{eqn:lambda_2_general}. The expression for the inter-delivery exponent $\lambda_1$ of user $U_1$ is same as \eqref{eqn:lambda_2_general} with $b$ and $c$, and $q_1$ and $q_2$ interchanged.

 \begin{figure}[t]
\centering
\includegraphics[width=3.3in]{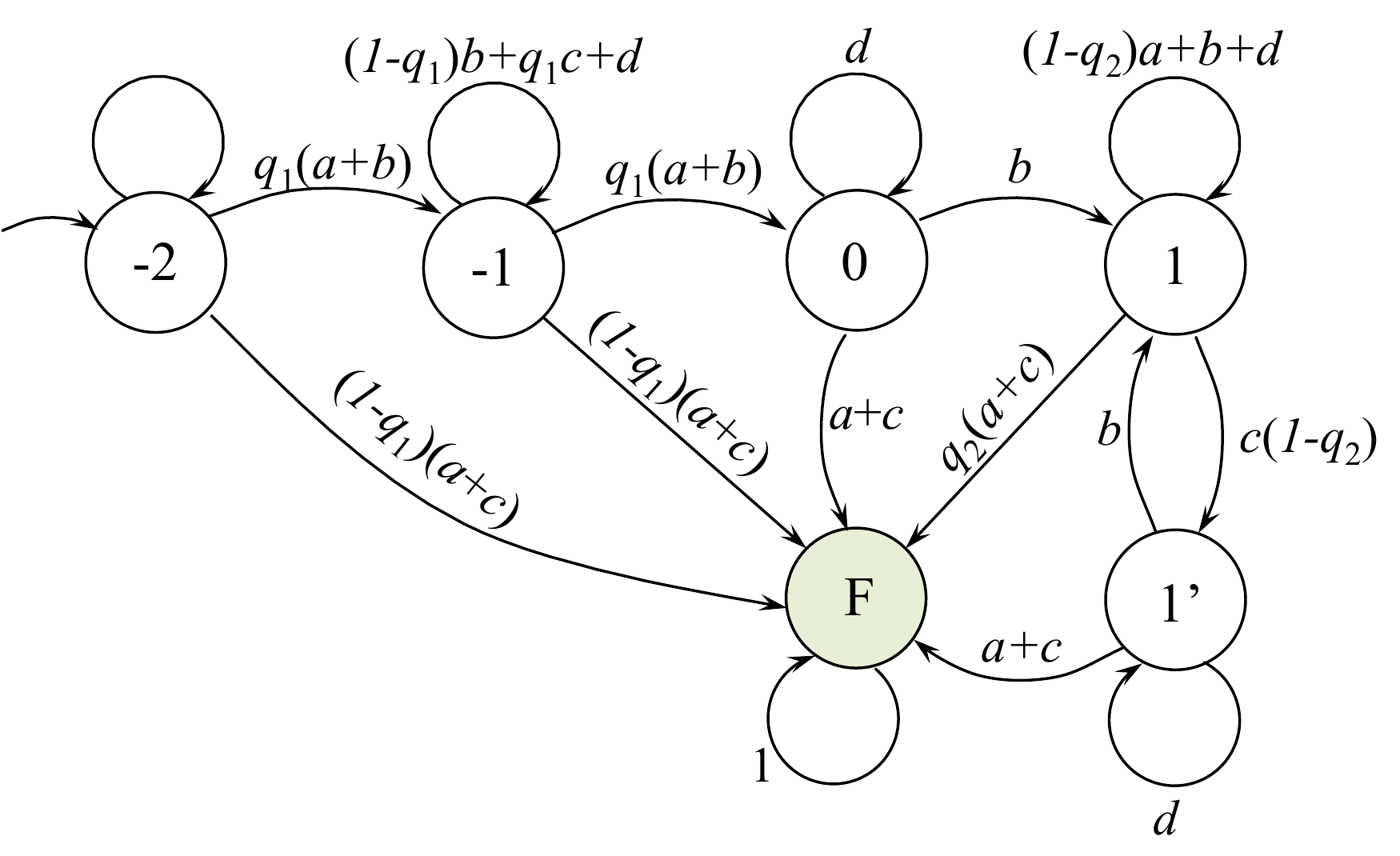}
\caption{Markov model used to determine the inter-delivery exponent $\lambda_2$ of user $U_2$. The absorbing state $F$ is reached when an in-order packet is decoded by $U_2$. The exponent of the distribution of the time taken to reach this state is $\lambda_2$. \label{fig:stab_markov_absorb_2}}
\end{figure}

\end{proof}

\end{appendix}
\bibliographystyle{IEEEtran}
\bibliography{streaming}

\end{document}